\documentclass[runningheads]{llncs}
\usepackage{graphicx}
\usepackage{subcaption}
\usepackage{amsmath}
\usepackage{amssymb}
\usepackage{xspace} 
\usepackage{multirow}
\usepackage{etoolbox}
\usepackage{thm-restate}
\usepackage{hyperref}
\usepackage{local}

\begin{document}

\title{ACORN: Network Control Plane Abstraction using Route Nondeterminism}
\author{Divya Raghunathan\inst{1} \and
Ryan Beckett\inst{2} \and
Aarti Gupta\inst{1} \and
David Walker\inst{1}}

\institute{
Princeton University \and
Microsoft Research}
\maketitle 

\begin{abstract}
Networks are hard to configure correctly, and misconfigurations occur frequently, leading to outages or security breaches. Formal verification techniques have been applied to
guarantee the correctness of network configurations, thereby improving network reliability. This work addresses verification of distributed network control planes, with two distinct contributions to improve the scalability of formal verification. Our first contribution is a hierarchy of abstractions of varying precision which introduce nondeterminism into the route selection procedure that routers use to select the best available route. We prove the soundness of these abstractions and show their benefits. Our second contribution is a novel SMT encoding which uses symbolic graphs to encode all possible stable routing trees that are compliant with the given network control plane configurations. We have implemented our abstractions and SMT encodings in a prototype tool called \sysname. Our evaluations show that our abstractions can provide significant relative speedups (up to 323x) in performance, and \sysname can scale up to $\approx37,000$ routers (organized in FatTree topologies, with synthesized shortest-path routing and valley-free policies) for verifying reachability. This far exceeds the performance of existing tools for control plane verification.

\end{abstract}

\section{Introduction}
\label{s:intro}
Bugs in configuring networks can lead to expensive outages or critical security breaches, and misconfigurations occur frequently~\cite{fastly,cloudflare,time-warner,twitter,bank,united}. Thus, there has been great interest in formal verification of computer network configurations. Many initial efforts targeted the network \emph{data plane}, \ie the forwarding rules in each router that determine how a given packet is forwarded 
to a destination. 
Many of these methods 
have been successfully applied in large data centers in practice~\cite{rcdc,jingjing,google-libra}. 
In comparison, formal verification of the network \emph{control plane} is more challenging. 

Traditional control planes use distributed protocols such as OSPF, BGP, and RIP~\cite{nwbook} to compute a network data plane based on the route announcements received from peer networks, the current failures detected, and the policy in configurations set by network administrators. In control plane verification,
one must check that \emph{all} data planes that emerge due to the router configurations are correct.
There has been much recent progress in control plane verification. 
For small-sized networks, fully symbolic SMT-based verifiers~\cite{minesweeper,bagpipe,nvPldi20} usually work well and support a broad range of properties. For medium-to-large networks, SMT-based verifiers have not been shown to scale well. Instead, simulation-based verifiers~\cite{batfish,fastplane,shapeshifter,nvPldi20,plankton,hoyan} work much better, and many use additional symbolic analyses in limited scenarios (\eg failures). However, in general, they do not provide full symbolic reasoning, \eg for considering all external route announcements.

Our work is motivated by this gap: we aim to provide full symbolic reasoning \emph{and} improve the scalability of verification.
We address this challenge with two main contributions -- a novel hierarchy of control plane abstractions, and a new symbolic graph-based SMT encoding for control plane verification.

\para{Hierarchy of nondeterministic abstractions} 
Our novel control plane abstractions introduce nondeterminism in the procedure that routers use to select a route  
-- we call these the \emph{Nondeterministic Routing Choice (NRC) abstractions}. Instead of forcing a router to pick the \emph{best} available route, we allow it to \emph{nondeterministically choose} a route from a subset of available routes which includes the best route. 
The number of non-best routes in this set determines the precision of the abstraction; our least precise abstraction corresponds to picking \emph{any} available route.
We formalize these abstractions using the \emph{Stable Routing Problem (SRP)} model~\cite{minesweeper,bonsai}, which can model a wide variety of distributed routing protocols.
Although some other efforts~\cite{bagpipe,lightyear} have also proposed to abstract the decision process in BGP (details in \S\ref{s:related}), 
we elucidate and study the general principle, prove it sound, and reveal a range of precision-cost tradeoffs.

Our main insight here is that 
the determination of a best route 
is not essential for verification of many correctness properties such as reachability (\eg when the number of hops may not matter). For such properties, it is often adequate to abstract the route selection procedure that computes the best route. Furthermore, for policy-based properties, routes that must not be taken are usually explicitly prohibited by network configurations via route filters, which we model accurately.
We show that our proposed abstractions are \emph{sound}
for verification of properties that must hold over all stable network states,
\ie they will not miss any property violations.

Although our abstractions are sound for verification under specified failures (\S\ref{s:verif_nrc}), in this paper we focus on symbolic verification using the SRP model without failures, which provides an important core capability. Our ideas can be combined with other prior techniques~\cite{minesweeper,hoyan,origami,nvPldi20} to consider at-most-k failures, which is outside the scope of this paper.

The potential downside of considering non-best routes 
is that our abstractions may lead to false positives, \ie we could report property violations although the \emph{best} route may actually satisfy the property. In such cases, we propose 
using a more precise abstraction that models more of the route selection procedure.
Our experiments (\S\ref{s:eval}) demonstrate that the \nrcs
can successfully verify a wide range of networks and common policies, and offer significant performance and scalability benefits in symbolic SMT-based verification.

\para{Symbolic graph-based SMT encoding} 
Our novel SMT encoding uses \emph{symbolic graphs}~\cite{monosat}, where 
Boolean variables are associated with each edge in the network topology, to model the stable routing trees that emerge from the network control plane. Our encoding can leverage specialized SMT solvers such as MonoSAT~\cite{monosat} that provide support for graph-based reasoning, in addition to standard SMT solvers such as Z3~\cite{z3}.

 \para{Experimental evaluation} 
 We have implemented our \nrcs and symbolic graph-based SMT encoding in a prototype tool called \sysname, and present a detailed evaluation on benchmark examples that include synthetic data center examples with FatTree topologies~\cite{fattree}, as well as real topologies from Topology Zoo~\cite{topozoo} and BGPStream~\cite{bgpstream}, running well-known network policies. These benchmarks, including some new examples that we created, are publicly available~\cite{benchmark-repo}.
 
 Our evaluations show that \sysname can verify reachability in large FatTree benchmarks with around $37,000$ nodes (running 
common policies) within an hour. 
This kind of scalability is needed in modern data centers with several thousands of routers that run distributed routing protocols such as BGP~\cite{bgp-datacenter}.
\emph{To the best of our knowledge, no other control plane verifier has shown the correctness of benchmarks of such large sizes.}

We also performed controlled experiments to evaluate the effectiveness of our abstractions. All benchmark examples were successfully verified using an NRC abstraction: 96\% of examples could be verified using our least precise abstraction, and the remaining 4\% were verified 
 using a more precise abstraction. 
For verifying reachability,
we can achieve relative speedups (compared with no abstraction) of up to 323x, and can scale to much larger networks as compared with no abstraction (our tool cannot scale beyond 4,500 routers without abstraction, but can verify a network with 36,980 routers using our least precise NRC abstraction). We compared \sysname with two publicly available state-of-the-art control plane verifiers on the FatTree benchmarks -- our results show that ShapeShifter~\cite{shapeshifter} and NV~\cite{nvPldi20} do not scale beyond 3,000 routers for verifying reachability (\S\ref{s:compare}).

To summarize, we make the following contributions:
\begin{enumerate}
\item We present a hierarchy of novel control plane abstractions (\S\ref{s:treeabs}), called the \nrcs, that introduce nondeterminism into a general route selection procedure.
We prove that our abstractions are sound
for verification, 
and empirically show that they enable a precision-cost tradeoff. Although our focus is on SMT-based verification, our  abstractions can be used with other methods as well.
\item We present a novel 
SMT encoding (\S\ref{s:symgraph}) (based on symbolic graphs~\cite{monosat}) to capture distributed control plane behavior which leverages specialized SMT solvers that support graph-based reasoning as well as standard SMT solvers.
\item  We have implemented our abstractions and SMT encoding in a prototype tool called \sysname, and present a detailed
evaluation (\S\ref{s:eval}) on benchmark examples~\cite{benchmark-repo} including synthetic data center networks~\cite{fattree} and real-world topologies from Topology Zoo~\cite{topozoo} and BGPStream~\cite{bgpstream}, running well-known network policies. 
\end{enumerate}

\section{Motivating Examples}
\label{s:motiv}
In a distributed routing protocol, routers exchange \emph{route announcements} containing information on how to reach various destinations. On receiving a route announcement, a router updates its internal state 
and then sends a route announcement to neighboring routers after 
processing it as per the routing configurations.
In well-behaved networks, this distributed decision process converges to 
 a \emph{stable state}~\cite{sppGriffin} in which the internal routing 
 information of each router does not change upon receiving additional route 
 announcements from its neighbors. The best route announcement selected by each router
 defines a \emph{routing tree}: if router $u$ chooses the route announcement 
 sent by router $v$ for destination $d$, then $u$ will forward data packets 
 with destination $d$ to $v$.

\begin{example}[Motivating example]
Consider the example network in Figure~\ref{fig:popl20}  
(from ShapeShifter~\cite{shapeshifter}) with five routers running the Border Gateway Protocol (BGP), where actions taken by routers are shown 
along the edges. 
\end{example}
The verification task is to check whether routes announced at $R_1$ can reach $R_5$.
The network is configured so that 
$R_4$ prefers to route through $R_3$: the community tag $c1$ is 
added on the edge from $R_1$ to $R_3$, which causes the local preference (lp) 
to be updated to 200 on the edge from $R_3$ to $R_4$. Routes with higher 
local preference values are preferred (the default local preference in 
BGP configurations is 100).
Thus, the best route at $R_4$ is through $R_3$ rather than through $R_2$, and the corresponding routing tree is shown by red (solid) arrows.

\begin{figure}
    \centering
\begin{subfigure}{0.49\textwidth}
  \centering
     \includegraphics[width=\textwidth]{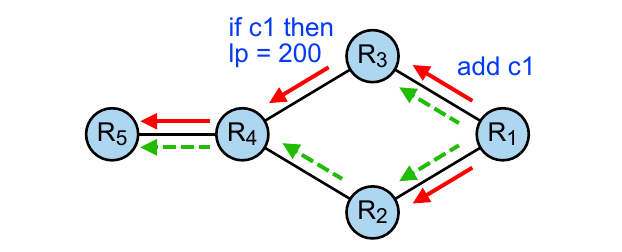}
    \caption{Example 1}
    \label{fig:popl20}
\end{subfigure}
\begin{subfigure}{0.49\textwidth}
    \includegraphics[width=\textwidth]{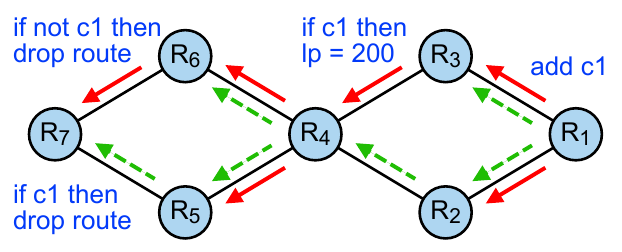}
    \caption{Example 2}
    \label{fig:diamond_example}
\end{subfigure}
\caption{Examples showing correct verification result with the \nca. Red arrows show the routing tree in the real network, and green arrows show an additional routing tree allowed in the \nca.}
\end{figure}

Note that $R_5$ can receive a route even if $R_4$ \emph{chooses} 
to route through $R_2$, though this route is not the best for 
$R_4$ (due to lower local preference). Thus, regardless of the choice 
$R_4$ makes, we can conclude that $R_5$ can reach the destination. 
This observation captures the basic idea in our \nrcs -- 
intuitively, we explore \emph{multiple} available routes at a node: the best route 
as well as other routes.
Then we check if $R_5$ receives a route announcement under 
\emph{each of these possibilities}. 
Since $R_5$ can reach $R_1$ in all 
routing trees considered by our abstraction (regardless of the choices made at $R_2$, $R_3$, or $R_4$), 
we correctly conclude that it can reach $R_1$.

\para{False positives and refinement}
The \nrcs are sound, \ie when verification with an NRC abstraction is successful, the property is guaranteed to hold in the network. We prove the soundness of our abstractions and show that they over-approximate network behavior (\S\ref{s:verif_nrc}). 
However, we could report a false positive, \ie a property violation even when the network satisfies the property. In the same Figure~\ref{fig:popl20}, suppose now that $R_5$ drops route announcements that do not have the tag $c1$. 
In the real network, $R_5$ will receive a route, since the route sent by $R_4$ has the tag $c1$. However, a verification procedure using the NRC abstraction which considers all possible routes at each router 
would report that $R_5$ \emph{cannot} reach the destination, with a counterexample where $R_4$ routes through $R_2$ and its route announcement is later dropped by $R_5$.
In this case, an NRC abstraction higher up in the hierarchy (with higher precision), \eg which chooses a route with maximum local preference but abstracts the path length, will verify that $R_5$ receives a route, thereby eliminating the false positive.
 
\para{Path-sensitive reasoning} 
Even the least precise NRC abstraction can verify many interesting policies due to 
our symbolic SMT-based approach which tracks correlations between 
choices made at different routers, which 
other tools~\cite{shapeshifter} do not track.
\begin{example}[Motivating example with path-sensitivity]
Figure~\ref{fig:diamond_example} shows another BGP network 
(based on an example from Propane~\cite{propane}), with seven routers and destination $R_1$. We would like to verify that $R_7$ can reach the destination.
\end{example}
In the real network, 
$R_4$ chooses the route from $R_3$ which has higher local preference, resulting in the routing tree shown by red (solid) arrows in the figure. Under the \nca which considers all available routes,
$R_4$ could choose the route from $R_2$ instead. Note that regardless of the choice made by $R_4$ the community tags in the routes advertised by $R_5$ and $R_6$ will be the same, and so $R_7$ will receive a route either way -- our abstraction 
is precise enough to
track this correlation 
and correctly concludes that $R_7$ can reach $R_1$.

\section{Preliminaries}
\label{s:prelim}
In this section we briefly cover the background on the key building blocks required to describe our technical contributions. Our abstraction of the control plane is developed on top of the Stable Routing Problem (SRP) model~\cite{minesweeper,bonsai,nvPldi20}, a formal model of network routing that can model distributed routing protocols such as BGP, OSPF, RIP, etc. 
We also briefly describe SMT-based verification using the SRP model (\eg Minesweeper~\cite{minesweeper}) and theory solvers for graphs used in the SMT solver MonoSAT~\cite{monosat}.

\begin{figure}[t]
\fbox{\begin{minipage}{\textwidth}
\textbf{SRP instance: } $\mathnormal{SRP = (G,\ A,\ a_d,\ \prec,\ \mathrm{trans}),\ G = (V,\ E,\ d)}$\\
\textbf{Properties of well-formed SRPs:}
\begin{align*}
&\forall v.\ (v, v) \notin E \tag{\emph{self-loop-free}}\\
&\forall e.\ \mathrm{trans}(e, \infty) = \infty  \tag{\emph{non-spontaneous}}
\vspace*{-0.2in}
\end{align*}
\textbf{SRP solution:} $\mathcal{L}: V \rightarrow \mathnormal{A_{\infty}}$
\begin{align*}
    \mathcal{L}(u) = \begin{cases}
    \mathnormal{a_d} & \text{ if } u = d\\
    \infty & \text{ if } \mathrm{attrs}_{\mathcal{L}}(u) = \emptyset\\
    a \in \mathrm{attrs}_{\mathcal{L}}(u) \text{, minimal by } \prec & \text{ if } \mathrm{attrs}_{\mathcal{L}}(u) \neq \emptyset
    \end{cases}
\end{align*}
\vspace*{-0.2in}
\begin{align*}
   & \mathrm{attrs}_{\mathcal{L}}(u) = \{a \ |\ (e, a) \in \mathrm{choices}_{\mathcal{L}}(u)\}\\
    &\mathrm{choices}_{\mathcal{L}}(u) = \{(e, a) \ |\  e = (v, u),\ a = \mathrm{trans}(e, \mathcal{L}(v)),\ a \neq \infty)\}
\end{align*}
\end{minipage}}
\caption{Cheat sheet for the SRP model~\cite{bonsai}.}
\label{fig:srp}
\end{figure}

\begin{definition}[Stable Routing Problem (SRP)]
An SRP is a tuple $(G,\ A,\ a_d,\ \prec,\ \mathrm{trans})$ where:
\begin{itemize}
    \item $G = (V,\ E,\ d)$ is a graph representing the network topology with vertices $V$, directed edges $E$, and destination vertex $d$. 
    \item $A$ is a set of \emph{attributes} that represent route announcements.
    \item $a_d \in A$ 
    represents the initial route announcement sent by the destination $d$.
    \item $\prec \ \subseteq A \times A$ is a partial order that models the route selection procedure that routers use to select the best route. If $a_1 \prec a_2$ then $a_1$ is preferred.
    \item $trans: E \times A_{\infty} \to A_{\infty}$, where $A_\infty = A \cup \{\infty\}$ and $\infty$ is a special value denoting the absence of a route, is a \emph{transfer function} that models the 
    processing 
    of route announcements sent from one router to another.
\end{itemize}
\end{definition}
Figure~\ref{fig:srp} summarizes the important notions for the SRP model~\cite{bonsai}. 
The main difference from routing algebras~\cite{sppGriffin,routingalgebra,metarouting} is that the SRP model includes a network topology graph $G$ to reason about a given network and its configurations. 
Note that the SRP model does not directly consider failures, but they can be accommodated using prior techniques~\cite{minesweeper,hoyan,origami,nvPldi20}. For example, Minesweeper~\cite{minesweeper} introduces additional Boolean variables $bi$ to model link/device failures, with $bi$ true indicating link (or device) $i$ has failed.

\begin{example}[SRP example]
The network in Figure~\ref{fig:popl20} running a simplified version of BGP (simplified for pedagogic reasons) is modeled using an SRP in which attributes are tuples comprising a 32-bit integer (local preference), a set of 16-bit integers (community tags), and a list of vertices (the path).
We use $a.lp$, $a.comms$, and $a.path$ to refer to the elements of an attribute $a$. The initial attribute at the destination, $a_d = (100, \emptyset, [\ ])$. The preference relation $\prec$ models the BGP route selection procedure which is used to select the best route. The attribute with highest local preference is preferred; to break ties, the attribute with minimum path length is preferred (more details are in Appendix~\ref{app:bgp}). The transfer function for edge $(R_1, R_3)$ adds the tag $c1$ and prepends $R_1$ to the path, returning $(100, a.comms \cup \{c1\}, [R_1] + a.path)$.
The transfer function for edge $(R_3, R_4)$ sets the local preference to $200$ if the tag $c1$ is present, \ie if $c1 \in a.comms$ it returns $(200, a.comms, [R_3] + a.path)$; otherwise, it returns $(100, a.comms,  [R_3] + a.path)$.
The transfer function for other edges $(u, v)$ prepends $u$ to the path, sets the local preference to the default value ($100$), and propagates the community tags.
\end{example}

\para{SRP solutions} A solution of an SRP is a labeling function $\mathcal{L}: V \rightarrow A_{\infty}$ which represents the final route (attribute) chosen by each node when the protocol converges. An SRP can have multiple solutions, or it may have none. Any SRP solution satisfies a \emph{local stability condition}: each node selects the best among the route announcements received from its neighbors.

\para{SMT-based network verification using SRP} In Minesweeper~\cite{minesweeper}, the SRP instance for the network is represented using an SMT formula $N$, such that satisfying assignments of $N$ correspond to SRP solutions. To verify if a network satisfies a given property encoded as an SMT formula $P$, the satisfiability of the formula $F = N \land \neg P$ is checked. If $F$ is satisfiable, a property violation is reported. Otherwise, the property holds over the network (assuming $N$ is satisfiable; if $N$ is unsatisfiable there are no stable paths in the network). 
We use the same overall framework:
we use an SMT formula $\wh{N}$ to encode an \emph{abstract} SRP instance, and check the  satisfiability of $\wh{N} \land \neg P$.

\para{SMT with theory solver for graphs}
MonoSAT \cite{monosat} is an SMT solver with support for \emph{monotonic} predicates. A predicate $p$ is (positive) monotonic in a variable $u$ if whenever $p(\ldots u = 0 \ldots)$ is true, $p(\ldots u = 1 \ldots)$ is also true.
Graph reachability is a monotonic predicate:
if node $v_1$ can reach node $v_2$ in a graph with some 
edges removed, then $v_1$ can still reach $v_2$ when the edges are added. MonoSAT leverages monotonicity to provide efficient theory support for graph-based reasoning using 
a \emph{symbolic graph}, a graph with a Boolean variable per edge. 
Formulas can include these Boolean edge variables as well as monotonic graph predicates 
such as reachability and max-flow.
MonoSAT has been used 
to check reachability in data planes in AWS networks~\cite{AwsCav19,bayless2021debugging} but not \emph{control planes}, as in this work.

\section{NRC Abstractions}
\label{s:treeabs}
 We formalize our \nrcs as \emph{abstract} SRP instances, each of which is parameterized by a partial order. 
 
\begin{definition}[Abstract SRP]
For an SRP $S = (G,\ A,\ a_d,\ \prec,\ trans)$, an abstract SRP $\wh{S}_{\prec'}$ 
is a tuple $(G,\ A,\ a_d,\ \prec',\ trans)$, where $G$, $A$, $a_d$, and $trans$ are defined as in the SRP $S$, and $\prec' \  \subseteq A_{\infty} \times A_{\infty}$ is a partial order which satisfies the following condition:
\begin{align}
    \forall B \subseteq A,\ \mathrm{minimal}(B, \prec) \subseteq \mathrm{minimal}(B, \prec')
    \label{eq:prec_cond}
\end{align}
where $\mathrm{minimal}(B, \prec) = \{a \in B \ | \ \nexists a' \in B.\ a' \neq a \land a' \prec a \}$ denotes the set of minimal elements of $B$ according to $\prec$.
\end{definition}

Condition~\eqref{eq:prec_cond} specifies that for any set of attributes $B$, the minimal elements of $B$ by $\prec$ are also minimal by $\prec'$. For example, let $\prec$ be the lexicographic ordering over pairs of integers: $(a, b) \prec (c, d)$ iff $a < c$ or $a = c$ and $b < d$. A partial order $\prec'$ that only compares the first components (\ie $(a, b) \prec' (c, d)$ iff $a < c$) satisfies condition~\eqref{eq:prec_cond}. 
Note that condition~\eqref{eq:prec_cond} ensures that the solutions (\ie minimal elements) at any node in an SRP are also solutions at the same node in the abstract SRP, \ie the \nrcs over-approximate the behavior of an SRP. 

The precision of the \nrcs varies according to the partial order used in the abstract SRP.
Our least precise abstraction uses $\prec^*$, the partial order in which any two attributes are incomparable and $\infty$ is worse than all attributes. The corresponding abstract route selection procedure chooses \emph{any} available route. The following example illustrates solutions of an abstract SRP $\wh{S}_{\prec ^*}$.

\begin{figure}[t]
\begin{center}
\begin{subfigure}[b]{0.49\textwidth}
     \centering
     \includegraphics[width=\textwidth]{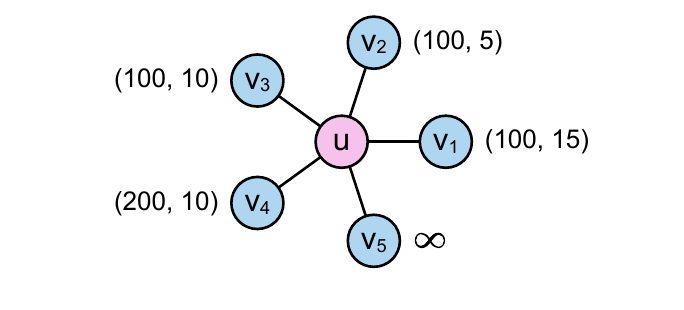}
      \caption{Attributes at node $u$'s neighbors}
     \label{fig:hasse1}
 \end{subfigure}
 \begin{subfigure}[b]{0.49\textwidth}
     \centering
    \includegraphics[width=\textwidth]{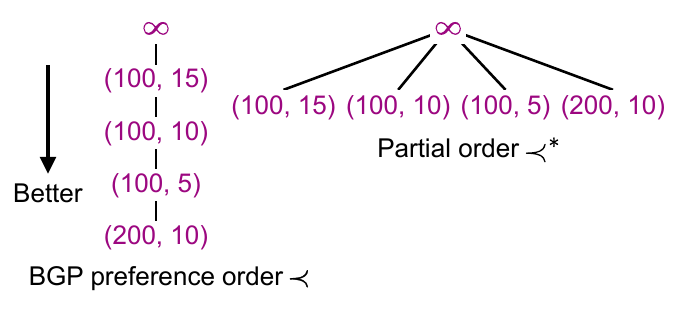}
     \caption{Partial orders in concrete (left) and abstract (right) SRPs}
     \label{fig:hasse2}
 \end{subfigure}
  \end{center}
  \caption{Example: Partial orders on neighbor attributes.}
  \label{fig:hasse}
\end{figure}

\begin{example}[Abstract SRP $\wh{S}_{\prec ^*}$] Consider the example shown in 
 Figure~\ref{fig:hasse1}, which shows simplified BGP attributes (pairs of local preference and path length) along with $\infty$ (which denotes no route), at the neighbors of a node $u$ in some network. Figure~\ref{fig:hasse2} shows the Hasse diagrams for the partially ordered sets 
comprising these attributes for two partial orders: (1) $\prec$, the  partial order in the standard (concrete) SRP (lifted to $A_\infty$) that models BGP's route selection procedure (described in Appendix~\ref{app:bgp}), shown on the left, and (2) $\prec^*$, the partial order corresponding to choosing any available route (defined above), shown on the right. Attributes appearing lower in the Hasse diagram are considered better. Hence, in the concrete SRP, $u$ will select $(200, 10)$. In the abstract SRP, any element that is minimal by $\prec ^*$ can be a part of a solution, so $u$ nondeterministically selects one of the routes it receives. Observe that $(200, 10)$, the solution for $u$ in the concrete SRP, is \emph{guaranteed} to be one of the solutions for $u$ in the abstract SRP. This over-approximation due to condition~\eqref{eq:prec_cond} ensures that our abstraction is sound, \ie it will not miss any property violations.
\end{example}

\para{Verification with the NRC abstractions}
\label{s:verif_nrc}
To verify that a property holds in a given network using an abstraction $\prec'$, we construct an SMT formula $\wh{N}$
such that satisfying assignments of $\wh{N}$ are solutions of the abstract SRP $\wh{S}_{\prec '}$ for the network, and conjoin it with the negation of an encoding of the  property $P$ to get a formula $F = \wh{N} \land \neg P$. If $F$ is unsatisfiable, then all abstract SRP solutions satisfy the property and verification is successful.
Otherwise, we report a violation and return the satisfying assignment.

Our approach is sound for properties that hold for all stable states of a network, \ie properties of the form $\forall \mathcal{L} \in Sol(S).P(\mathcal{L})$, where $Sol(S)$ denotes the SRP solutions for the network. 
Note that, like Minesweeper~\cite{minesweeper}, our approach only models the stable states of a network, and cannot verify properties over transient states that arise before the routing protocol converges.

\begin{restatable}{lemma}{overapproxlemma}[Over-approximation]
\label{lemma:overapprox}
For an SRP $S$ and a corresponding abstract SRP $\wh{S}_{\prec '}$ with solutions $Sol(S)$ and  $Sol(\wh{S}_{\prec '})$ respectively, $Sol(S) \subseteq Sol(\wh{S}_{\prec'})$.
\end{restatable}
\noindent The proof follows from the definition of SRP solutions and the over-approximation ensured by  condition~\eqref{eq:prec_cond} (the complete proof is in Appendix~\ref{app:proofs}). 

\begin{restatable}{theorem}{soundnessthm}[Soundness]
\label{thm:soundness}
Given SMT formulas $\wh{N}$ and $N$ modeling the abstract and concrete SRPs, respectively, and an SMT formula $P$ encoding the property to be verified, if $\wh{N} \land \neg P$ is unsatisfiable, then $N \land \neg P$ is also unsatisfiable.
\end{restatable}
\noindent The proof follows directly from Lemma~\ref{lemma:overapprox} and is presented in Appendix~\ref{app:proofs}.

\para{Verification under failures} 
We model link failures using $\infty$, which denotes no route (device failures are modeled by failures of all incident links). Let $F$ denote the set of failed links. Given SRP $S = (G, A, a_d \prec, \mathrm{trans})$, we model network behavior under failures of edges in $F$ using an SRP $S_{F} = (G, A, a_d, \prec, \mathrm{trans}_F)$ where $\mathrm{trans}_F$ models failed links by returning $\infty$ along edges in $F$, and is the same as $\mathrm{trans}$ for other edges. To 
model failures of links in $F$ 
using our abstraction, we define the abstract SRP for $S_F$, $\wh{S}_{\prec' F} = (G, A, a_d, \prec', \mathrm{trans}_F)$; it only differs from $S_F$ in the partial order $\prec'$. Since Lemma~\ref{lemma:overapprox} holds for an arbitrary concrete SRP $S$, it holds for $S_F$, \ie any solution of $S_F$ is also a solution of $\wh{S}_{\prec' F}$. Hence, the \nrcs are sound for verification under specified failures.

\para{Hierarchy of NRC abstractions} The least precise NRC abstraction uses the partial order $\prec ^*$ defined above, which corresponds to choosing any available route. More precise abstractions can be obtained by modeling the route selection procedure \emph{partially} rather than choosing any route nondeterministically.
Figure~\ref{fig:gen_nrc} shows partial orders and the corresponding route selection procedures (presented as steps in a ranking function) for OSPF and BGP, ordered from the least precise ($\prec ^*$) to the most precise ($\prec$). For example, $\prec_{(lp, pathlen)}$ corresponds to the first two steps of BGP's route selection procedure (described in Appendix~\ref{app:bgp}), \ie it is a two-step ranking function which first finds routes with maximum local preference, and from these routes, selects one with minimum path length. Abstractions higher up in the hierarchy are more precise (\ie result in fewer false positives) as they model more of the route selection procedure, but are more expensive as their SMT encodings involve more variables and constraints. This tradeoff between precision and performance in the \nrcs is evident in our experiments: verification with the $\prec_{(lp)}$ partial order was successful on all networks for which verification with $\prec ^*$ gave false positives (\S\ref{s:eval_wan}), but took up to 2.7x more time than verification with $\prec ^*$.
\begin{figure}
    \centering
\begin{tabular}{|l|c|p{8cm}|}
\hline
Protocol & Partial order & Best route (lexicographic ordering)\\
\hline
\multirow{3}{*}{OSPF} & $\prec ^*$ & Any (nondeterministic choice)\\
& $\prec_{path cost}$ & min path cost\\
& $\prec_{ospf}$ & min path cost, min router id\\
\hline
\multirow{5}{*}{BGP} & $\prec ^*$ & Any (nondeterministic choice)\\
& $\prec_{(lp)}$ & max lp (local preference)\\
& $\prec_{(lp, pathlen)}$ & max lp, min path length\\
& $\prec_{(lp, pathlen, MED)}$ & max lp, min path length, min MED\\
& $\prec_{bgp}$ & max lp, min path length, min MED, min router ID\\
\hline
\end{tabular}
    \caption{Hierarchy of \nrcs for OSPF and BGP Protocols.}
    \label{fig:gen_nrc}
\end{figure}

\para{Abstraction refinement} 
If verification with a specific abstraction fails, we first \emph{validate} the reported counterexample in the concrete SRP by checking if each node actually chose the best route. Note that the selected routes in the counterexample may contain only some fields, depending on the abstraction used. We first find the values of the other fields and the set of available routes by applying the transfer functions along the edges in the counterexample, starting from the destination router (\ie by effectively simulating the counterexample on the concrete SRP). We then check if all routers selected the best route that they received. If this is the case, we have found a stable solution in the concrete SRP that violates the property; if not, the counterexample is \emph{spurious}. We can eliminate the spurious counterexample by adding a blocking clause (\ie the negation of the variable assignment corresponding to the counterexample), and repeat verification with the same abstraction in a CEGAR~\cite{cegar} loop. However, this could take many iterations to terminate. Instead, we 
suggest choosing a more precise abstraction which is higher up in the NRC hierarchy. We could potentially 
use a \emph{local} refinement procedure that uses a higher-precision abstraction only at certain routers, based on the counterexample. We plan to explore this and other ways of counterexample-guided abstraction refinement in future work.

\section{SMT Encodings}
\label{s:symgraph}
In this section we present SMT encodings for an abstract SRP based on symbolic graphs, as well as encodings of 
properties addressed in this paper, such as reachability and policy properties. 
We begin by providing definitions 
for a symbolic graph and its solutions.
\begin{definition}[Symbolic graph~\cite{monosat}]
A symbolic graph $\mathcal{G}_{RE}$ is a tuple $(G, RE)$ where $G = (V, E)$ is a graph representing the network topology and $RE = \{re_{uv} \ | \ (u, v) \in E\}$ is a set of Boolean routing edge variables.
\end{definition}

\begin{definition}[Symbolic graph solutions~\cite{monosat}] 
A symbolic graph $\mathcal{G}_{RE} = (G, RE)$ and a formula $F$ over $RE$ 
has solutions  $Sol(\mathcal{G}_{RE},\ F)$ which are subgraphs of $G$ defined by assignments to $RE$ that satisfy $F$, such that an edge $(u, v)$ is in a solution subgraph iff $re_{uv} = 1$ in the corresponding satisfying assignment.
\end{definition}

\subsection{Routing Constraints on Symbolic Graphs}
We now describe how to construct an SMT formula $\wh{N}$ 
such that the symbolic graph solutions $Sol(G_{RE}, \wh{N})$ correspond to solutions of the abstract SRP for the network, where $G_{RE}$ is the symbolic graph for the network topology $G$. Our formulation consists of constraints grouped into the following categories (the complete formulation is summarized in Figure~\ref{fig:encoding}):

\begin{itemize}
    \item \textbf{Routing choice constraints:} Each node other than the destination either chooses a route from a neighbor or chooses $None$ (which denotes no route). We use a variable $nChoice$ at each node to denote either the selected neighbor or $None$. The routing edge $re_{vu}$ is true if and only if $u$ chooses a route from $v$, \ie $nChoice_u = nID(u, v)$, where $nID(u, v)$ denotes $u$'s neighbor ID for $v$.
    \item \textbf{Route availability constraints:} If a node $u$ chooses the route from its neighbor $v$, then $v$ must have a route to the destination. 
    This ensures that symbolic graph solutions are trees rooted at the destination vertex. If all neighbors $v$ of $u$ either don't have a route ($\neg hasRoute_v$) or the route is dropped from $v$ to $u$ ($routeDropped_{vu}$), then $u$ must choose $None$.
    \item \textbf{Attribute transfer and route filtering constraints:} If $u$ selects the route from its neighbor $v$ (\ie $re_{vu} = 1$), the route should not be dropped along the edge $(v, u)$ and the transfer function must relate the attributes of $u$ and $v$, ensuring that symbolic graph solutions respect route filtering policies.
    The attribute at the destination is the initial route announcement $a_d$.
\end{itemize}

\begin{figure}
\fbox{
\begin{minipage}{\textwidth}
\textbf{Abstract SRP } $\mathnormal{\wh{S} = (G,\ A,\ a_d,\ \prec',\  trans),\ G = (V,\ E,\ d)}$\\
\textbf{Symbolic graph } $\mathnormal{\mathcal{G}_{RE} = (G,\ RE)}$\\\\
\textbf{Variables}\\
\begin{tabular}{p{5cm} p{6.5cm}}
$attr1_u, attr2_u, \ldots$ : bit vector & \emph{route announcement fields, $\forall u \in V$}\\
$nChoice_u$ : bit vector & \emph{neighbor choices,  $\forall u \in V \setminus \{d\}$}\\
$hasRoute_u$ : Boolean & \emph{placeholder for route availability, $\forall u \in V$}\\
$routeDropped_{uv}$ : Boolean & \emph{route dropped along an edge, $\forall (u, v) \in E$}
\end{tabular}\\\\
\textbf{Constants}\\
\begin{tabular}{p{5cm} p{6.5cm}}
$nID(u, v)$ : integer & \emph{$u$'s neighbor ID for $v$, $\forall (u, v) \in E$}\\
$None_u$ : integer & \emph{special ID denoting no neighbor, $\forall u \in V$}\\
\end{tabular}\\\\
\textbf{Routing choice constraints}
\begin{align}
    &\left( \bigvee_{(v, u) \in E} nChoice_u = nID(u, v) \right) \lor nChoice_u = None_u \label{eq:nchoice}\\
    &nChoice_u = nID(u, v) \leftrightarrow re_{vu} \label{eq:re_choice}\\
    &\neg re_{vd} \ \ \forall (v, d) \in E
    \label{eq:in_dest}
\end{align}
\textbf{Route availability constraints}
\begin{align}
    &nChoice_u = nID(u, v) \rightarrow hasRoute_v \label{eq:r_avail}\\
    &nChoice_u = None_u \leftrightarrow \nonumber\\
    &\bigwedge_{(v, u) \in E} \neg hasRoute_v \lor routeDropped_{vu}
    \label{eq:no_route}
\end{align}
\textbf{Attribute transfer and route filtering constraints}
\begin{align}
    &re_{vu} \rightarrow  attr_u = trans_{vu}(attr_v)
    \label{eq:trans}\\
    &re_{vu} \rightarrow \neg routeDropped_{vu}
    \label{eq:dropped}\\
    &attr_{d} = a_d
    \label{eq:init}
\end{align}
\textbf{Solver-specific constraints}\\\\
(a) SMT solvers with graph theory support (\eg MonoSAT):
\begin{align}
    &\forall u \in V, \ hasRoute_u \leftrightarrow \mathcal{G}_{RE}.reaches(d, u) \label{eq:graph_r_avail}
\end{align}
(b) SMT solvers without graph theory support (\eg Z3):
\begin{align}
    &hasRoute_{d} \label{eq:nograph_r_avail1}\\
    &\forall u \neq d,\ hasRoute_{u} \leftrightarrow \bigvee_{(v, u) \in E} hasRoute_v \land re_{vu}\label{eq:nograph_r_avail2}\\
    &rank_{d} = 0\label{eq:nograph_r_avail3}\\
    &\forall (v, u) \in E,\ re_{vu} \rightarrow  rank_u = (rank_v + 1) \label{eq:nograph_r_avail4}
\end{align}
\end{minipage}
}
\caption{SMT encoding for abstract SRP}
\label{fig:encoding}
\end{figure}

Our formulation is parameterized by three placeholders: $hasRoute_v$, which is true iff $v$ receives a route 
from the destination; $trans_{vu}$, the transfer function along edge $(v, u)$; and $routeDropped_{vu}$, which models route filtering along edge $(v, u)$. Of these, $trans_{vu}$ and $routeDropped_{vu}$ depend on the network protocol and configuration. 
(These are shown in detail in Appendix~\ref{app:example_transfer_constraints}
for the network in Figure~\ref{fig:diamond_example}.) The encodings of 
$hasRoute$ are described in the next subsection.

\subsection{Solver-specific Constraints} 
\label{s:solver-specific}
We have two encodings of the placeholder $hasRoute$, depending on whether the SMT solver has graph theory support. Additionally, we leverage graph-based reasoning to concisely encode transfer functions that use regular expressions over the AS path. This is a commonly used feature in BGP policies, but is not supported by most control plane verification tools.

\para{SMT solvers with graph theory support}

\noindent
\emph{Route availability placeholder.} We use the 
reachability predicate $\mathcal{G}_{RE}.reaches$ to encode  route availability:
$hasRoute_v$ (\ie node $v$ has a route to destination $d$) if and only if
$\mathcal{G}_{RE}.reaches(d, v)$ (\ie
there is a path from $d$ to $v$ in the symbolic graph $\mathcal{G}_{RE}$).

\noindent
\emph{Regular expressions over AS path.} We can support regular expressions that specify that the AS path contains certain ASes or sequences of ASes by using the \re variables and the reachability predicate without introducing any additional variables.
For example, the regular expression ``.*ab.*c.*d.*"  (where `.' matches any character and `*' denotes 0 or more occurrences of the preceding character) matches any path that traverses the edge (a, b) and then passes through nodes c and d. Existing SMT encodings would require additional variables to track a sequence of nodes in the AS path, but we can encode this concisely as: 
$re_{ab} \land \mathcal{G}_{RE}.reaches(b, c) \land \mathcal{G}_{RE}.reaches(c, d)$.

\para{Standard SMT solvers} For standard SMT solvers without specialized support for graph-based reasoning, such as Z3~\cite{z3}, we interpret the placeholder $hasRoute$ as a reachability marker which indicates whether the node has received a route announcement, and add constraints to ensure that the marker is suitably propagated in the symbolic graph. 
To prevent solutions with loops, we use another variable, $rank$, at each node to track the length of the path along with additional constraints (shown in Figure~\ref{fig:encoding}).

\para{Loop prevention} In BGP, routing loops are prevented by tracking the list of ASes (autonomous systems) in the path using the AS path attribute; routers drop route announcements if the AS path contains their AS. To model BGP's loop prevention mechanism exactly, 
Minesweeper's~\cite{minesweeper} SMT encoding 
would require $O(N^2)$ additional variables (where $N$ is the number of routers in the network) to track, for each router, the set of routers in the AS path. As this is expensive, Minesweeper~\cite{minesweeper} uses an optimization that relies on the route selection procedure to prevent loops when routers use the default local preference: the shorter loop-free path will be selected, as routes with shorter path length are preferred. Our encodings for $hasRoute$ model BGP's loop prevention mechanism exactly with much fewer than $O(N^2)$ additional variables: 
the MonoSAT encoding does not need any additional variables and our Z3 encoding uses $O(N)$ additional variables ($rank$) to prevent symbolic graph solutions with loops.

\subsection{Benefits of the \nrcs in SMT solving} 
\label{s:benefits}
The \nrcs provide the following benefits in improving SMT solver performance and scalability.

\para{Fewer attributes} 
The most direct benefit is that many route announcement fields 
become irrelevant and can be removed from the network model, resulting in smaller SMT formulas.
Specifically, all fields required to model \emph{route filtering} (\ie the dropping of route announcements) and the property of interest are retained, but fields that are used only for route selection (\eg local preference and path length) can be removed depending on the specific abstraction. 
    
\para{Expensive transfers can be avoided during SMT search}
Once a neighbor is selected during the SMT search, then transfers of attributes from other neighbors become irrelevant. An attribute transfer constraint has the $re$ variable on the left hand side of an implication~\eqref{eq:trans}. Thus, if $re_{vu}$ is assigned false (\ie $v$ is not selected) during the search, the constraint is trivially satisfied and the transfer function (on the right side) becomes irrelevant, thereby improving solver performance. In contrast, without any abstraction, 
each node 
\emph{must} consider transfers of attributes from \emph{all} neighbors 
to pick the best route.

\subsection{Encoding Properties for Verification}
\label{s:prop_encoding}
\para{Reachability} We encode the property that a node $u$ can reach destination $d$ by asserting its negation: $nChoice_u = None_u$.

\para{Non-reachability/Isolation} We encode the property that a node $u$ can never reach the destination $d$ by asserting its negation: $nChoice_u \neq None_u$.

\para{No-transit property} Routing policies between 
autonomous systems (ASes) are typically influenced by business relationships between them, such as provider-customer or peer-peer. A provider AS is paid to carry traffic to and from its customers, while peer ASes exchange traffic between themselves and their customers without any charge. The BGP policies (Gao-Rexford conditions~\cite{gao-rexford,bgp_policies}) between ASes usually ensure that an AS does not carry traffic from one peer or provider to another. 
This is called a \emph{no-transit} property; 
its negation is encoded as the constraint:
\begin{equation}
    \bigvee_{u \in V} \ \ \bigvee_{\substack{v, w \in PeerProv(u),\\v \neq w}} re_{vu} \land re_{uw}
\end{equation}
where $PeerProv(u)$ is the set of neighbors of $u$ that are its peers or providers.

\para{Policy properties} 
A BGP policy can be defined by assigning a particular meaning to specific community tags. Policy properties can then be checked by asserting a formula over the community attribute at a node. 
\begin{figure}[t]
\centering
\begin{subfigure}[b]{0.48\textwidth}
    \centering
      \includegraphics[width=\linewidth]{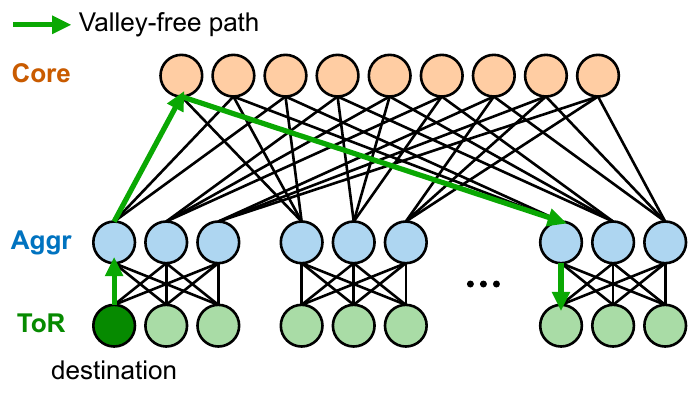}
      \caption{FatTree topology}
      \label{fig:valleyfree_topo}
\end{subfigure}
\begin{subfigure}[b]{0.48\textwidth}
    \centering
       \includegraphics[width=\linewidth]{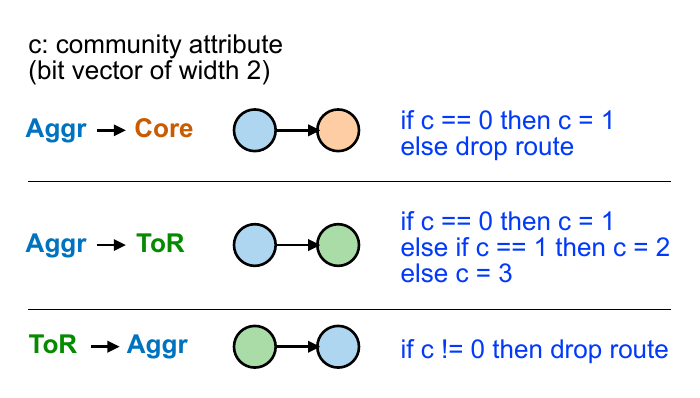}
      \caption{Valley-free policy}
      \label{fig:valleyfree_policy}
\end{subfigure}
\caption{Example data center network with a valley-free policy.}
\label{fig:valleyfree}
\end{figure}

\begin{example}[Valley-free policy in data center networks]
\label{ex:valley-free}
The valley-free policy prevents paths that have valleys, \ie paths which go up, down, and up again between the layers of a FatTree network topology~\cite{fattree,propane}; a valley-free path is highlighted (in green) in Figure~\ref{fig:valleyfree_topo} which shows a FatTree topology. The lower two layers of the FatTree are grouped into pods; three pods are shown in Figure~\ref{fig:valleyfree_topo}. An implementation of the valley-free policy using the community attribute $c$ in BGP is shown in Figure~\ref{fig:valleyfree_policy}. A path between ToR (top-of-rack) routers in different pods that contains a valley will pass through at least 3 Aggr (aggregation) routers and will have $c = 3$ (with the logic for updating $c$ as shown in Figure~\ref{fig:valleyfree_policy}). Hence, the valley-free property at a node $u$ is checked by asserting its negation: 
$comm_{u} = 3$.
\end{example}

\section{Implementation and Evaluation}
\label{s:eval}
We implemented our SMT encoding (\S\ref{s:symgraph}) for the abstract SRP with the \nrcs, 
and extended it for the concrete SRP 
using additional constraints (described in Appendix~\ref{app:concreteSmt}) in our prototype tool, \sysname. It uses the Python APIs of the MonoSAT and Z3 solvers. The tool input is an 
intermediate representation (IR) of a network topology and  configurations that 
represents routing policy using match-action rules\footnote{
Our IR and benchmarks are described in Appendix~\ref{app:grammar}, and are publicly available~\cite{benchmark-repo}.}, similar to the route-map constructs in Cisco's configuration language. (Our IR will serve as a target for front-ends such as Batfish~\cite{batfish} and NV~\cite{nvPldi20} in future work.) 

In our evaluation, we measure the effectiveness of the \nrcs and use two different back-end SMT solvers -- MonoSAT and Z3 (with bitvector theory and bit-blasting enabled). We consider the following four settings in our experiments:
\begin{enumerate}
\item \textbf{abs\_mono}: with \nca ($\prec ^*)$, using MonoSAT 
\item \textbf{abs\_z3}: with \nca ($\prec ^*)$, using Z3 
\item \textbf{mono}: without abstraction, using MonoSAT 
\item \textbf{z3}: without abstraction, using Z3.
\end{enumerate}

We evaluated \sysname on two types of benchmarks: (1) data center networks with FatTree topologies~\cite{fattree}, with sizes ranging from 125 to 36,980 routers (FatTree parameter ranging from k=10 to k=172), and (2) wide area networks comprising topologies from Topology Zoo~\cite{topozoo}, and a new set of benchmark examples that we created using parts of the Internet reported by BGPStream~\cite{bgpstream} to have been involved in misconfiguration incidents. More details of our benchmark examples are provided in Appendix~\ref{app:grammar}. 
We also compared \sysname with two state-of-the-art control plane verifiers on the data center benchmarks. In this comparison, the no-abstraction settings indicate the effectiveness of our symbolic graph-based SMT encoding, which is distinct from other SMT-based tools. All experiments were run on a Mac laptop with a 2.3 GHz Intel i7 processor and 16 GB memory. 

\subsection{Data Center Benchmarks}
\label{s:dc}
We generated data center benchmark examples with FatTree topologies~\cite{fattree} with well-known policies (used in prior work). These include: (1) a shortest-path routing policy, (2) the valley-free policy described in example~\ref{ex:valley-free}, (3) an extension of the valley-free policy that uses regular expressions to enforce isolation between a FatTree pod and an external router connected to the core routers of the FatTree, and (4) a buggy valley-free policy in which routers in the last pod cannot reach routers in other pods. We checked reachability between a router in the first pod and a router in the last pod for all policies, and a policy-based property for (2) and (3). The results are shown in Figure~\ref{fig:dc_results}, with each graph showing the number of nodes on the x-axis and the verification time (in seconds) on the y-axis.

\begin{figure}
    \centering
\begin{subfigure}{0.31\textwidth}
    \centering
     \includegraphics[width=\textwidth]{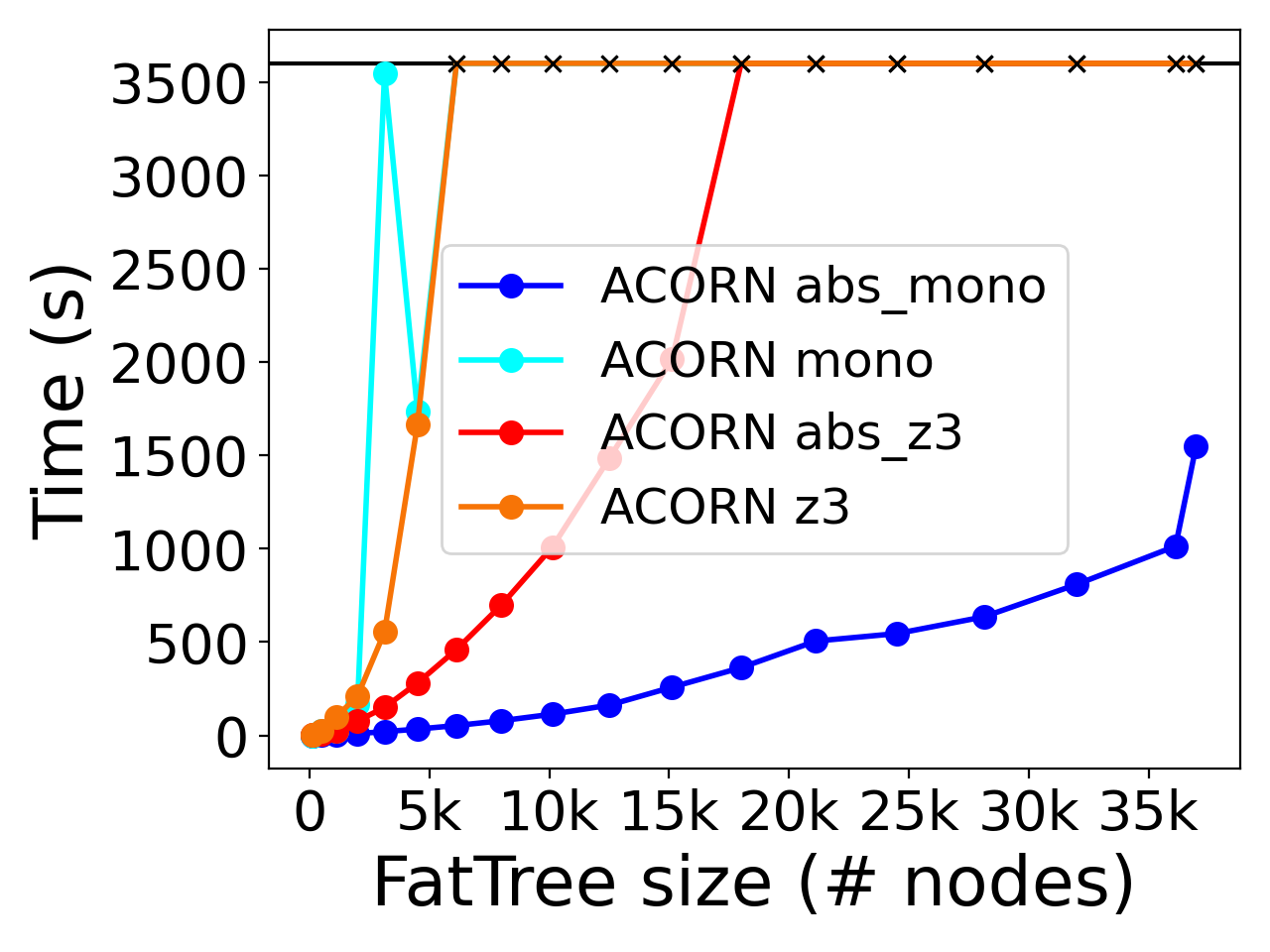}
    \caption{}
    \label{fig:sp_reach}
\end{subfigure}
\begin{subfigure}{0.31\textwidth}
    \centering
    \includegraphics[width=\textwidth]{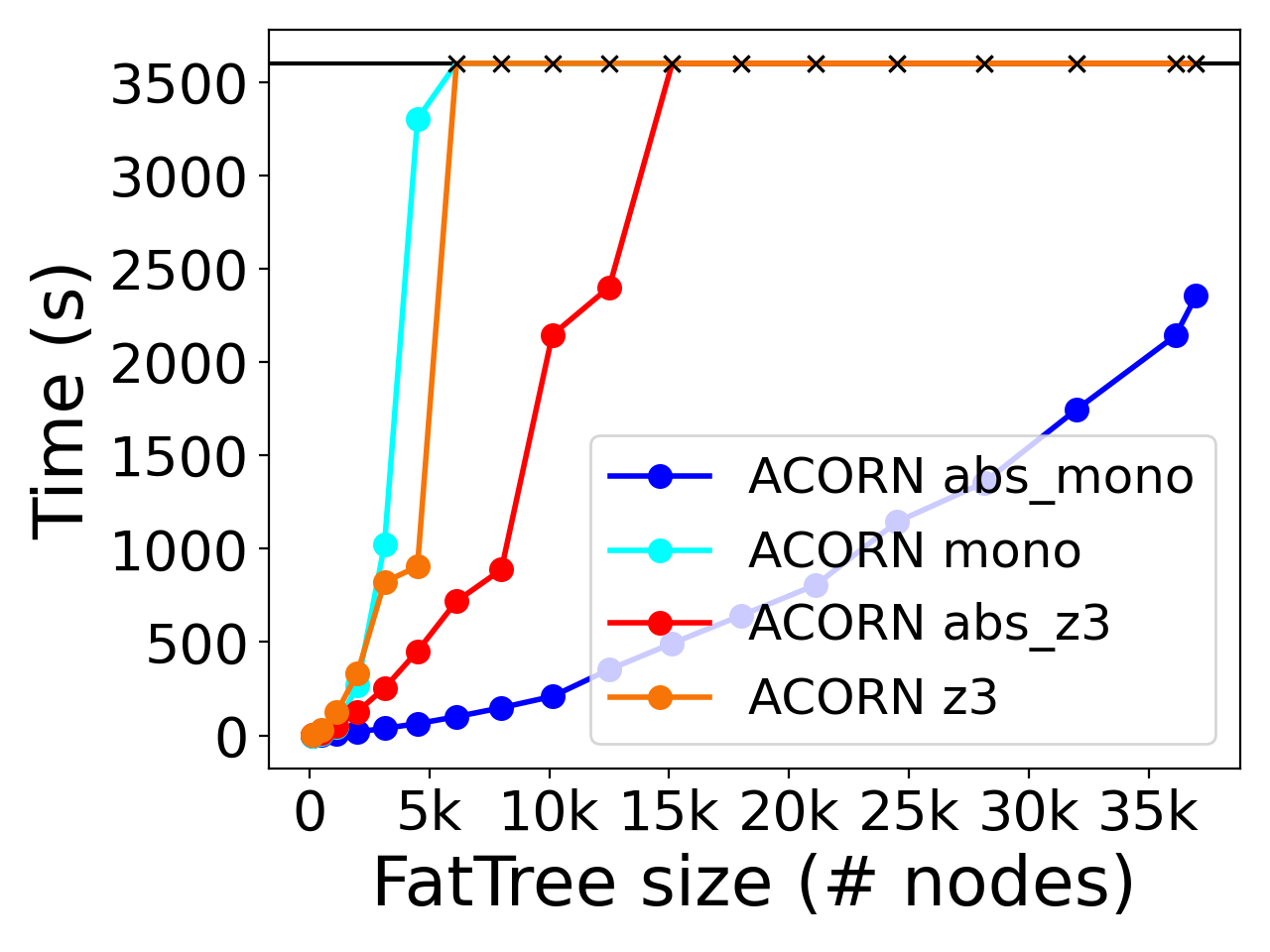}
    \caption{}
    \label{fig:vf_reach}
\end{subfigure}
\begin{subfigure}{0.31\textwidth}
    \centering
    \includegraphics[width=\textwidth]{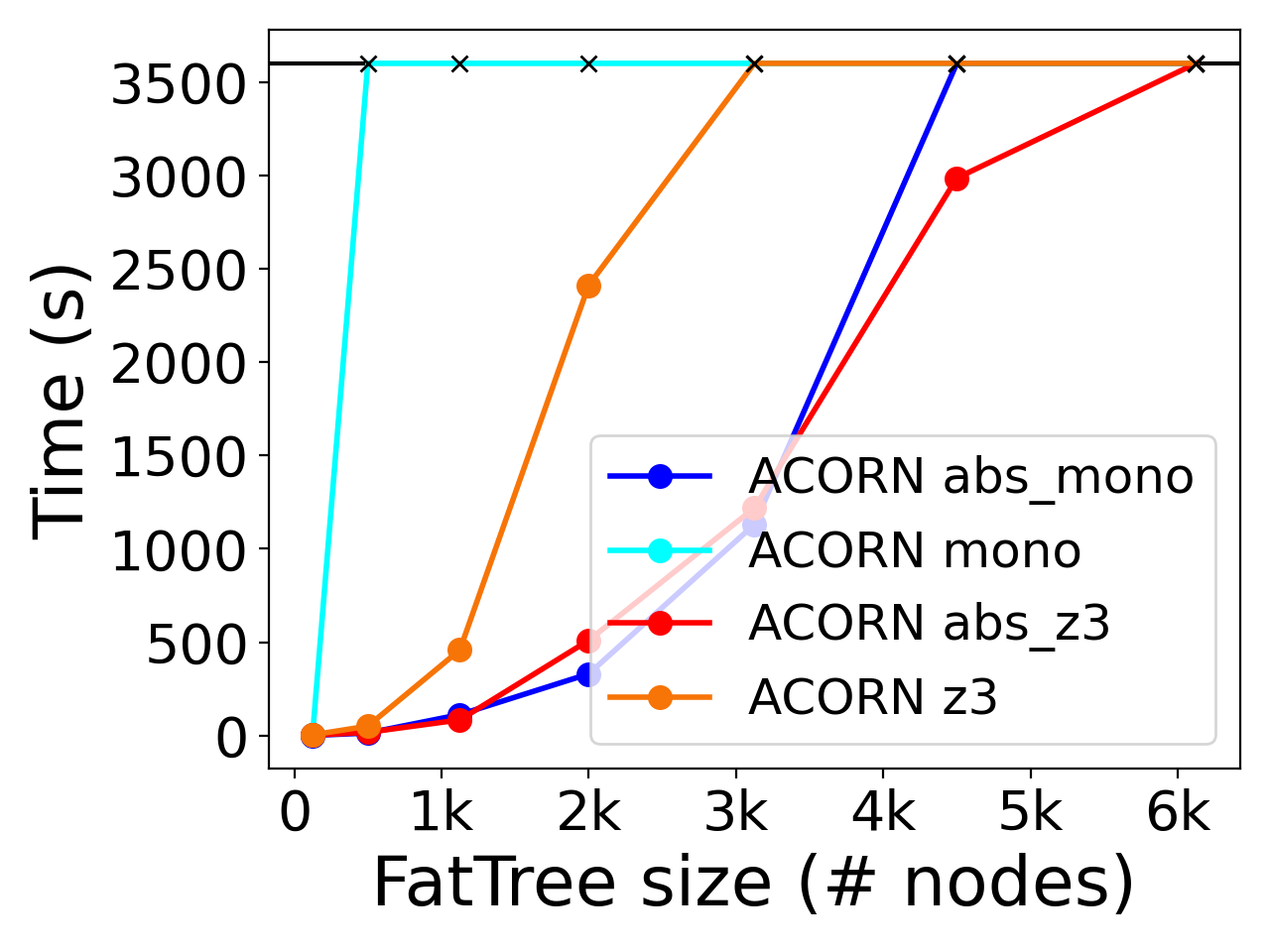}
    \caption{}
    \label{fig:vf_prop}
\end{subfigure}
\begin{subfigure}{0.31\textwidth}
    \centering
    \includegraphics[width=\textwidth]{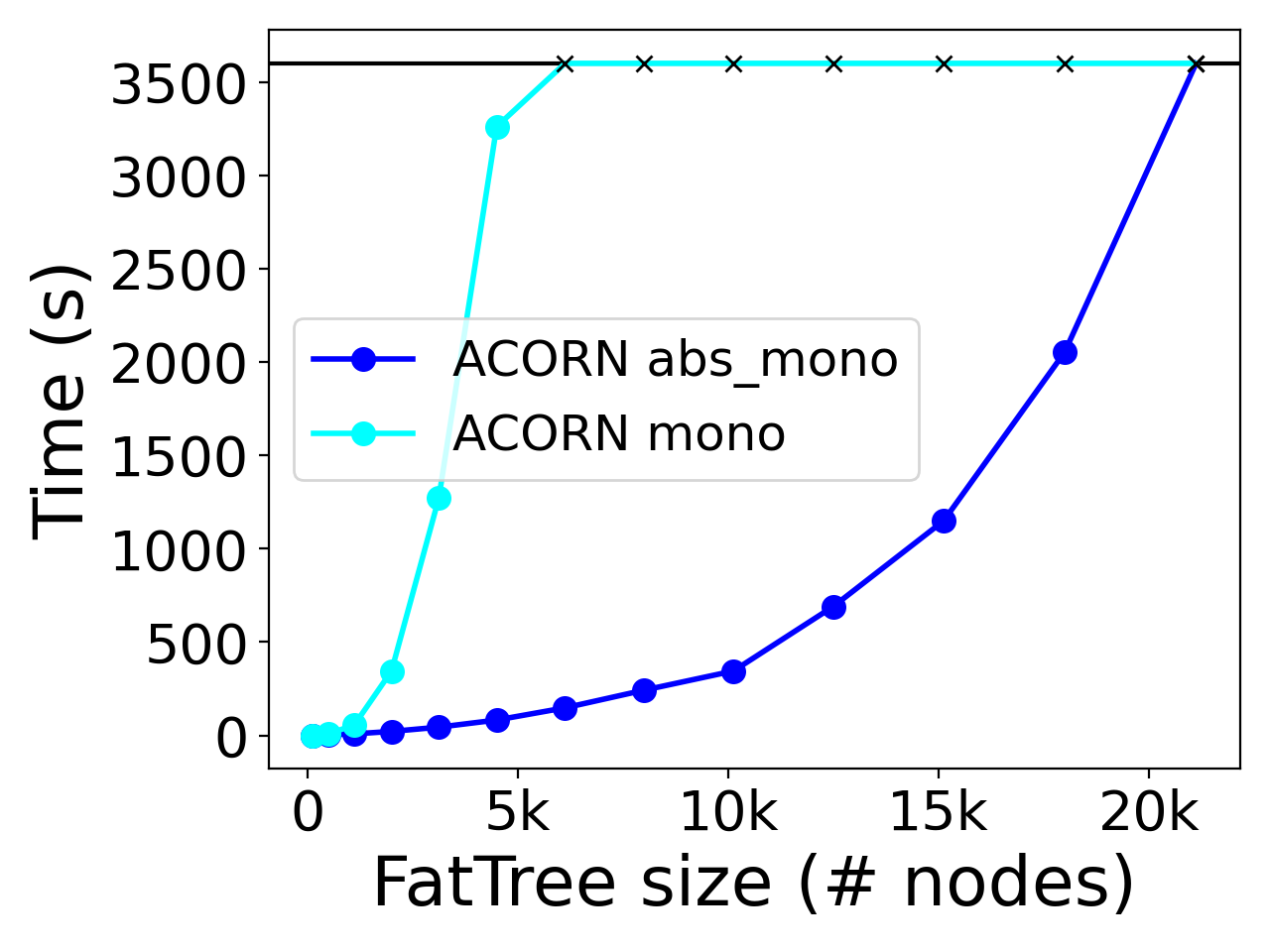}
    \caption{}
    \label{fig:ext_acc_reach}
\end{subfigure}
\begin{subfigure}{0.31\textwidth}
    \centering
    \includegraphics[width=\textwidth]{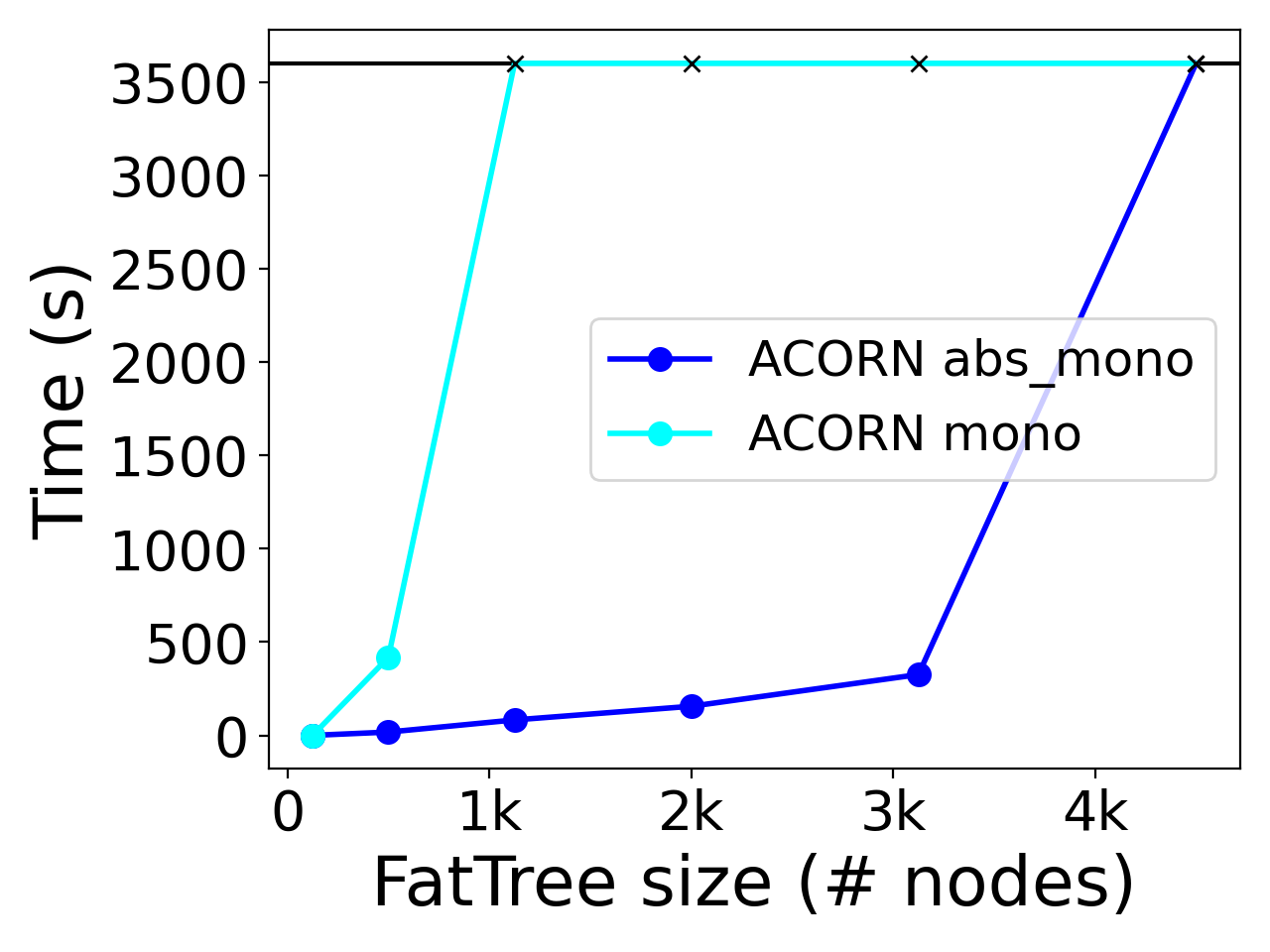}
    \caption{}
    \label{fig:ext_acc_isolation}
\end{subfigure}
\begin{subfigure}{0.31\textwidth}
    \centering
    \includegraphics[width=\textwidth]{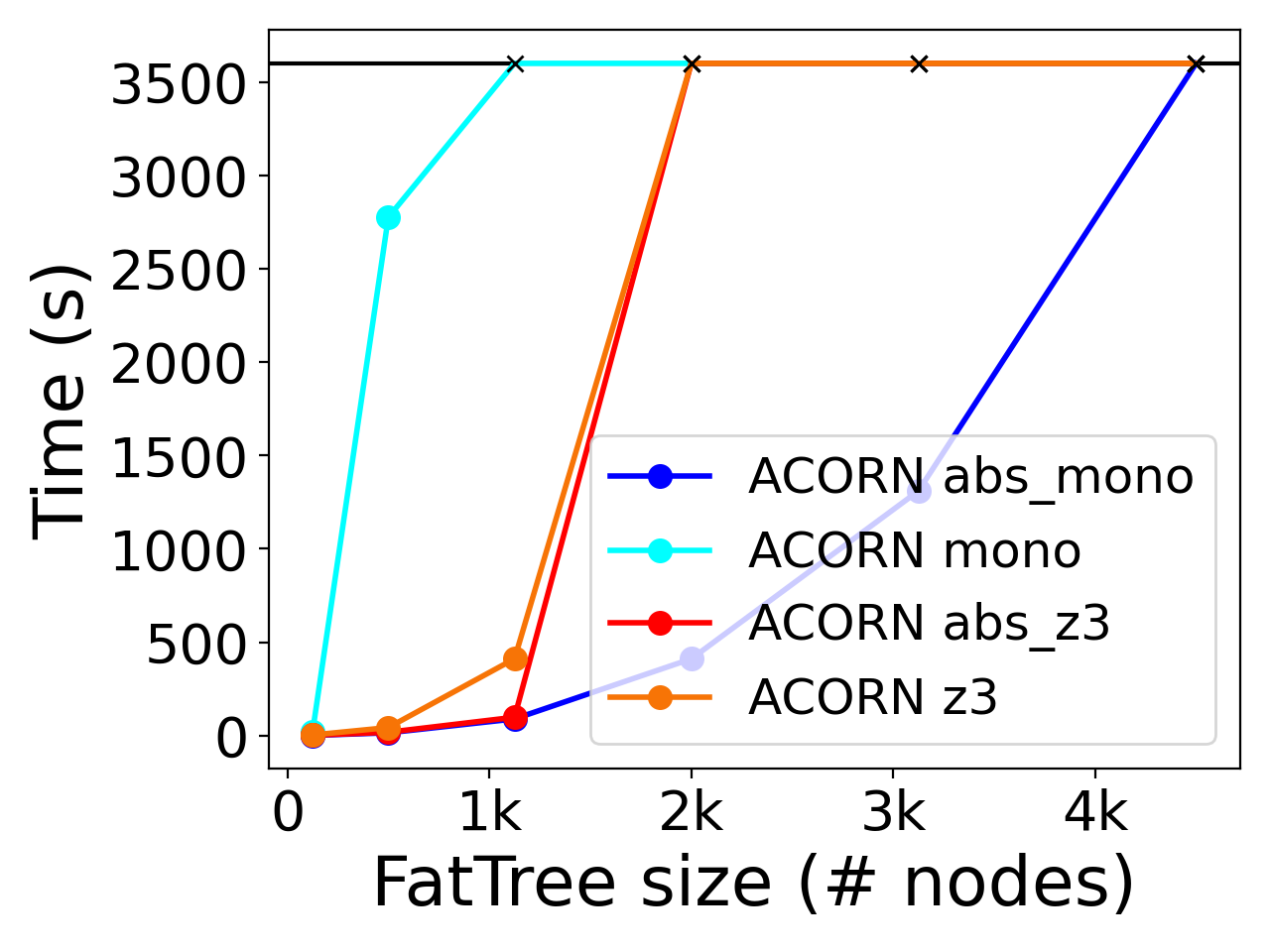}
    \caption{}
    \label{fig:buggy_vf_reach}
\end{subfigure}
    \caption{Results for data center network examples: (a) Reachability with shortest-path routing, (b) Reachability with valley-free policy, (c) Valley-free property, (d) Reachability with isolation policy, (e) Isolation property, and (f) Reachability with a buggy valley-free policy.}
    \label{fig:dc_results}
\end{figure}

\para{Reachability for shortest-path routing and valley-free policies} The results are shown in Figures~\ref{fig:sp_reach}, ~\ref{fig:vf_reach}, and~\ref{fig:ext_acc_reach} for the shortest-path routing policy, valley-free policy, and valley-free policy with isolation respectively. For both solvers, the abstract settings successfully verify reachability without any false positives and scale much better than the respective no-abstraction settings, showing that using the \nca is a clear win. For both policies, abs\_mono scales much better than abs\_z3 (we only ran the MonoSAT settings on the isolation policy examples as our encoding of regular expressions requires graph theory support). Our abs\_mono setting verifies reachability in a FatTree with 36,980 nodes running the valley-free policy in 40 minutes while abs\_z3 times out for a FatTree with 15,125 nodes. This shows the effectiveness of our symbolic graph-based SMT encoding that can leverage graph-based reasoning in MonoSAT for checking reachability, which is an important property that network operators care about. 

\para{Reachability for a buggy valley-free policy} We introduced a bug in the valley-free policy described earlier, where ToR routers in the last pod erroneously drop routes that do not have a community value of $0$. Since routes with a destination in another pod have a community value of $2$, these routes will get dropped. Our tool correctly reports that the destination is unreachable and provides a counterexample. The results are shown in Figure~\ref{fig:buggy_vf_reach}. Our abs\_mono setting can check a network with 3,000 nodes within an hour, while both the no-abstraction settings are worse, and cannot check a network with 2,000 nodes within an hour. With the MonoSAT solver, our abstract setting is up to $190$x faster than the no-abstraction setting, showing that our abstraction is effective even in cases when the network violates a reachability property.

\para{Policy properties}
For the valley-free policy examples, we checked the valley-free policy property (\S\ref{s:prop_encoding}). For the isolation policy examples, we checked isolation between a ToR router and the external router. The results are shown in Figures~\ref{fig:vf_prop} and~\ref{fig:ext_acc_isolation}, respectively. Again, the abstract settings verify both properties without false positives and outperform the respective no-abstraction settings, showing the benefit of using the \nca. For the valley-free policy property, where graph-based reasoning is likely not directly useful, abs\_z3 scales better than abs\_mono and can verify a network with 4,500 nodes within 50 minutes while abs\_mono times out. For both solvers, verifying reachability scales relatively better than verifying the policy-based properties.
\label{eval:dc}

\para{Summary of results} 
Our experiments show that for both solvers and for all properties -- reachability, as well as policy-based properties -- using the \nca gives better performance than using the no-abstraction setting which computes the best route at each node. In these benchmarks, our abstract settings successfully verify all properties without any false positives. The relative performance gains of the \nca are stronger for verifying reachability properties. In particular, with the MonoSAT solver, the \nca can achieve a relative speed-up of 52x for verifying reachability (when verification completes successfully within a 1 hour timeout for both abstract and non-abstract settings). Also, MonoSAT performed better than Z3 by up to 10x, demonstrating that it paid off in such cases to use SMT encodings that leverage graph-based reasoning. 
In terms of network size, for both solvers, the no-abstraction setting times out beyond 4,500 nodes for reachability verification, while the abstract setting scales up to about 37,000 nodes for the shortest-path and valley-free policies, and up to 18,000 nodes for the isolation policy with regular expressions. 

\subsection{Wide Area Networks}
\label{s:eval_wan}
To evaluate \sysname on less regular network topologies than data centers, we also conducted experiments on wide area network benchmarks. These typically have relatively small sizes and are not easily parameterized, unlike data center topologies. 
We considered two sets of benchmarks: (1) networks from Topology Zoo~\cite{topozoo}, which we annotated with business relationships, and (2) example networks based on parts of the Internet that were involved in misconfiguration incidents as reported on BGPStream~\cite{bgpstream}, with business relationships 
provided by the CAIDA AS relationships dataset~\cite{caida}. For both sets of benchmarks, 
we used the BGP policy shown in Figure~\ref{fig:gao_rex_pol}, which implements the Gao-Rexford conditions~\cite{gao-rexford} that guarantee BGP convergence: (1) routes from peers and providers are not exported to other peers and providers, and (2) routes from customers are preferred over routes from peers, which are preferred over routes from providers. We checked two properties: reachability of all nodes to a destination, and the no-transit property (\S\ref{s:prop_encoding}).

\para{Topology Zoo networks} We considered 10 of the larger examples from Topology Zoo~\cite{topozoo}, with sizes ranging from 22 to 79 routers. Since these describe only the topology, we created policies by annotating them with business relationships.
All settings take less than 0.5s for both properties (detailed results are in Appendix~\ref{app:grammar}). Note that due to smaller sizes, the runtimes are much faster than the runtimes for the data center benchmarks. Even here, the abstract settings are up to 3x faster than the respective no-abstraction settings, and successfully verify both properties without any false positives.

\para{Networks from BGPStream} 
We created 10 new benchmarks based on parts of the Internet in which BGPStream~\cite{bgpstream} detected possible BGP hijacking incidents, and used publicly available business relationships (CAIDA AS Relationships dataset~\cite{caida}). (This requires manual effort; we plan to create more benchmarks in future work.) 
The results are shown in Figures~\ref{fig:internet_reach} and~\ref{fig:internet_no_transit}, with the number of nodes (ASes) shown on the x-axis and verification time in seconds shown on the y-axis on a log scale. The abstract settings successfully verified reachability in 6 out of 10 benchmarks and reported false positives in 4 (indicated by triangular markers in Figure~\ref{fig:internet_reach}). Both abstract settings performed better than the respective no-abstraction settings, with relative speedups of up to 323x for MonoSAT and 3x for Z3 (when successful). Both abstract settings successfully verified the no-transit property for all networks, 
with abs\_mono performing much better (by up to 120x) than mono, while abs\_z3 performs better than z3 for some networks, but worse than z3 for others. It is clear from  Figure~\ref{fig:internet_reach} that the \nca significantly improves verification performance for both solvers for checking reachability.

\begin{figure}
\begin{subfigure}{0.3\textwidth}
    \centering
    \includegraphics[width=\linewidth]{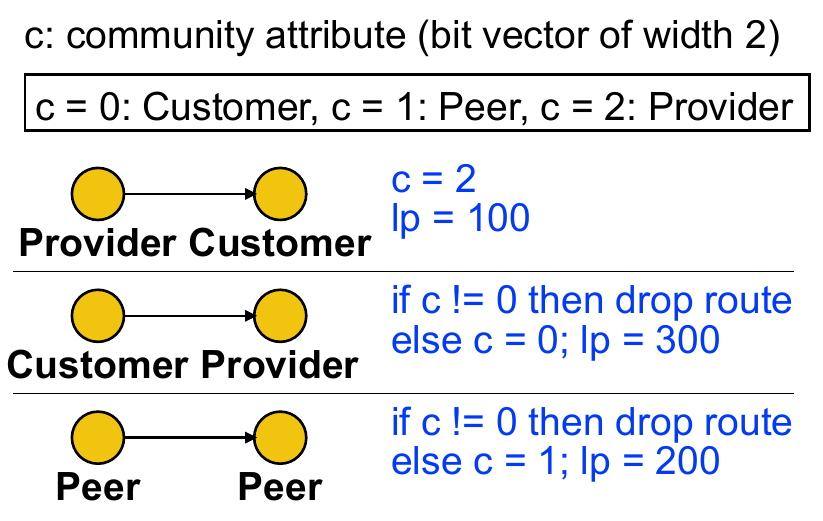}
    \caption{BGP policy}
    \label{fig:gao_rex_pol}
\end{subfigure}
\begin{subfigure}{0.34\textwidth}
        \centering
         \includegraphics[width=\linewidth]{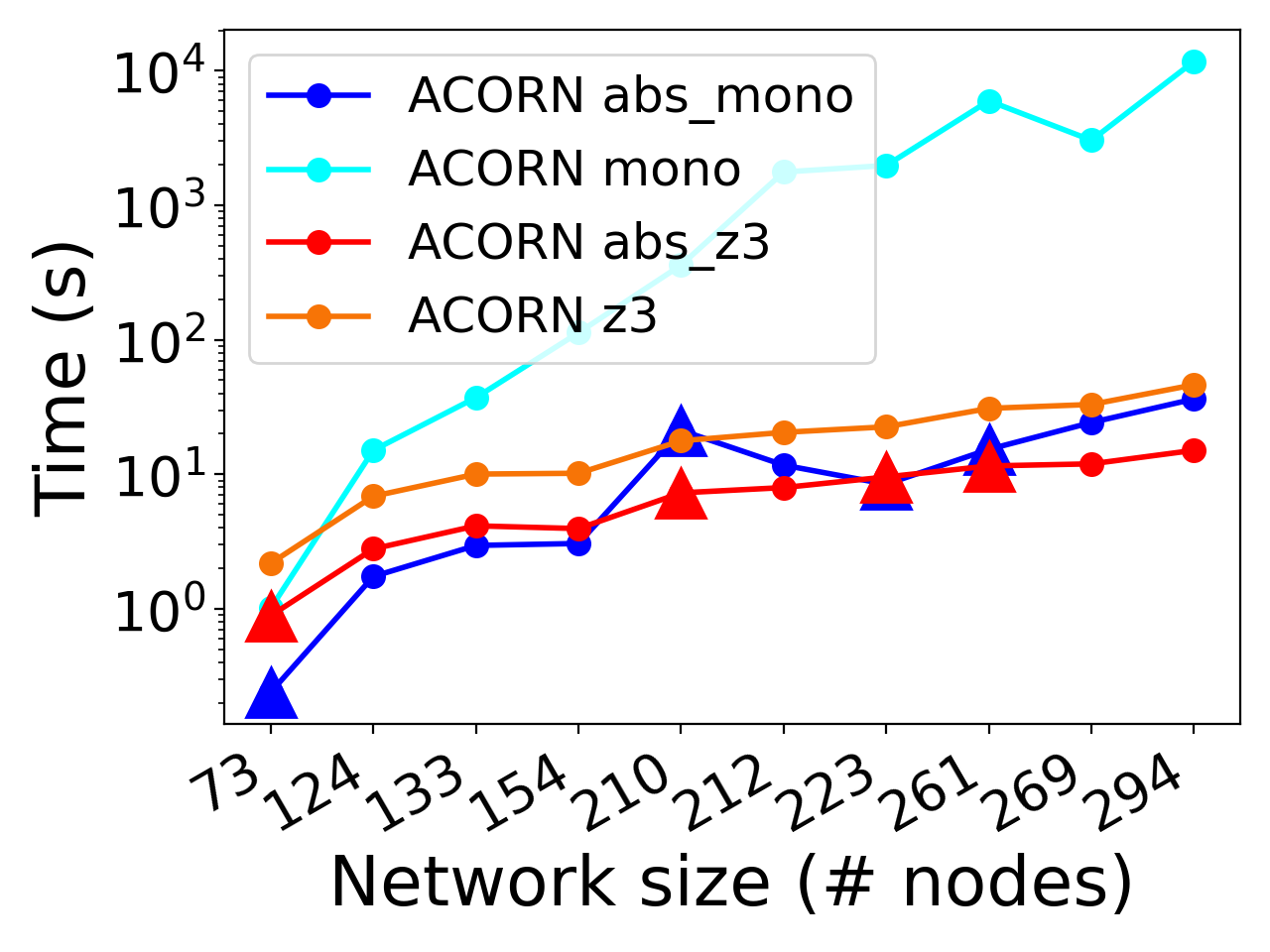}
        \caption{Reachability}
        \label{fig:internet_reach}
\end{subfigure}
\begin{subfigure}{0.34\textwidth}
    \centering
    \includegraphics[width=\linewidth]{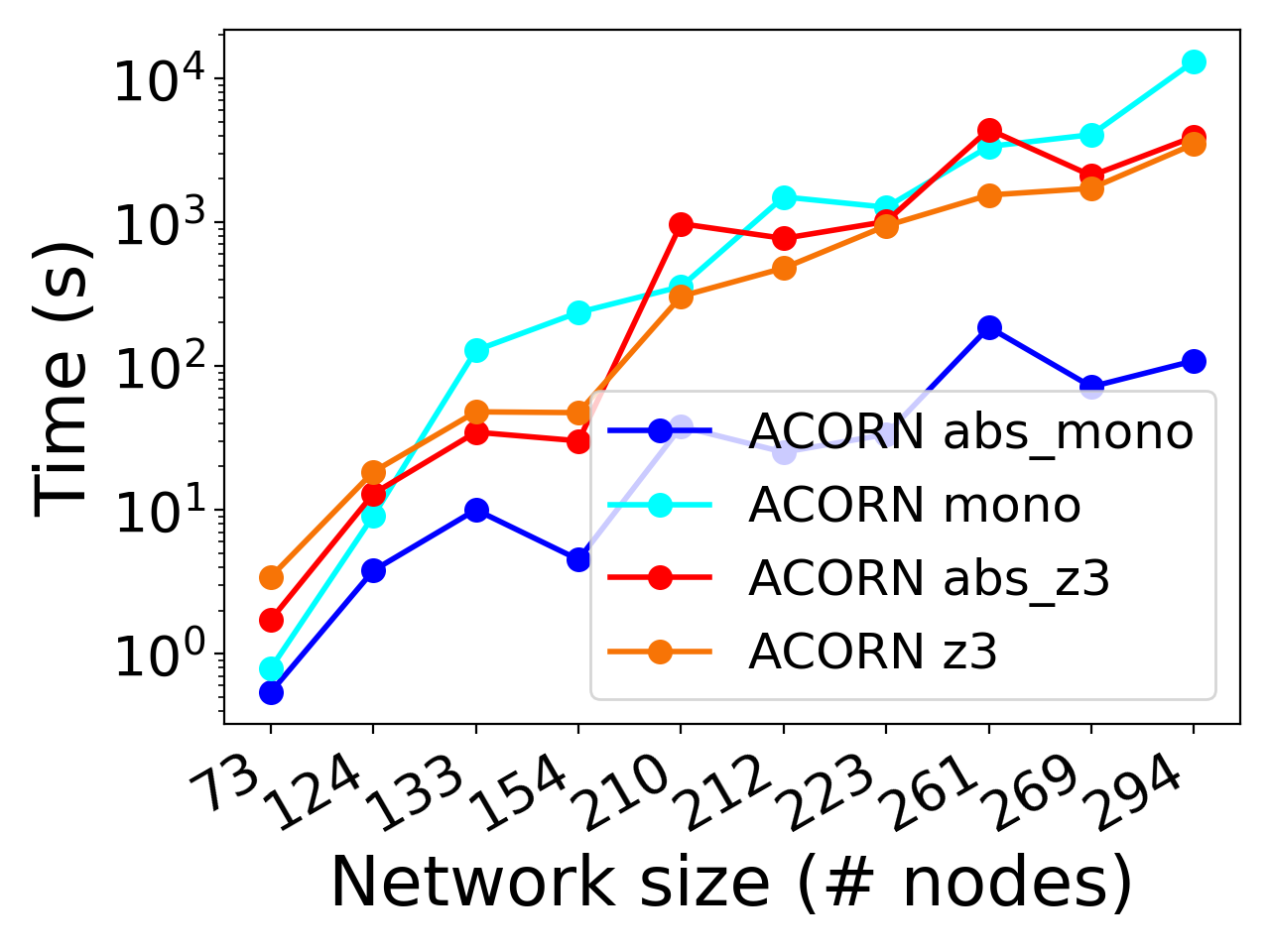}
    \caption{No-transit property}
    \label{fig:internet_no_transit}
\end{subfigure}
 \vspace*{-0.1in}
    \caption{BGP policy and results for wide area networks from BGPStream.}
\end{figure}

\begin{figure}
    \centering
    \begin{subfigure}{0.34\textwidth}
        \centering
         \includegraphics[width=\linewidth]{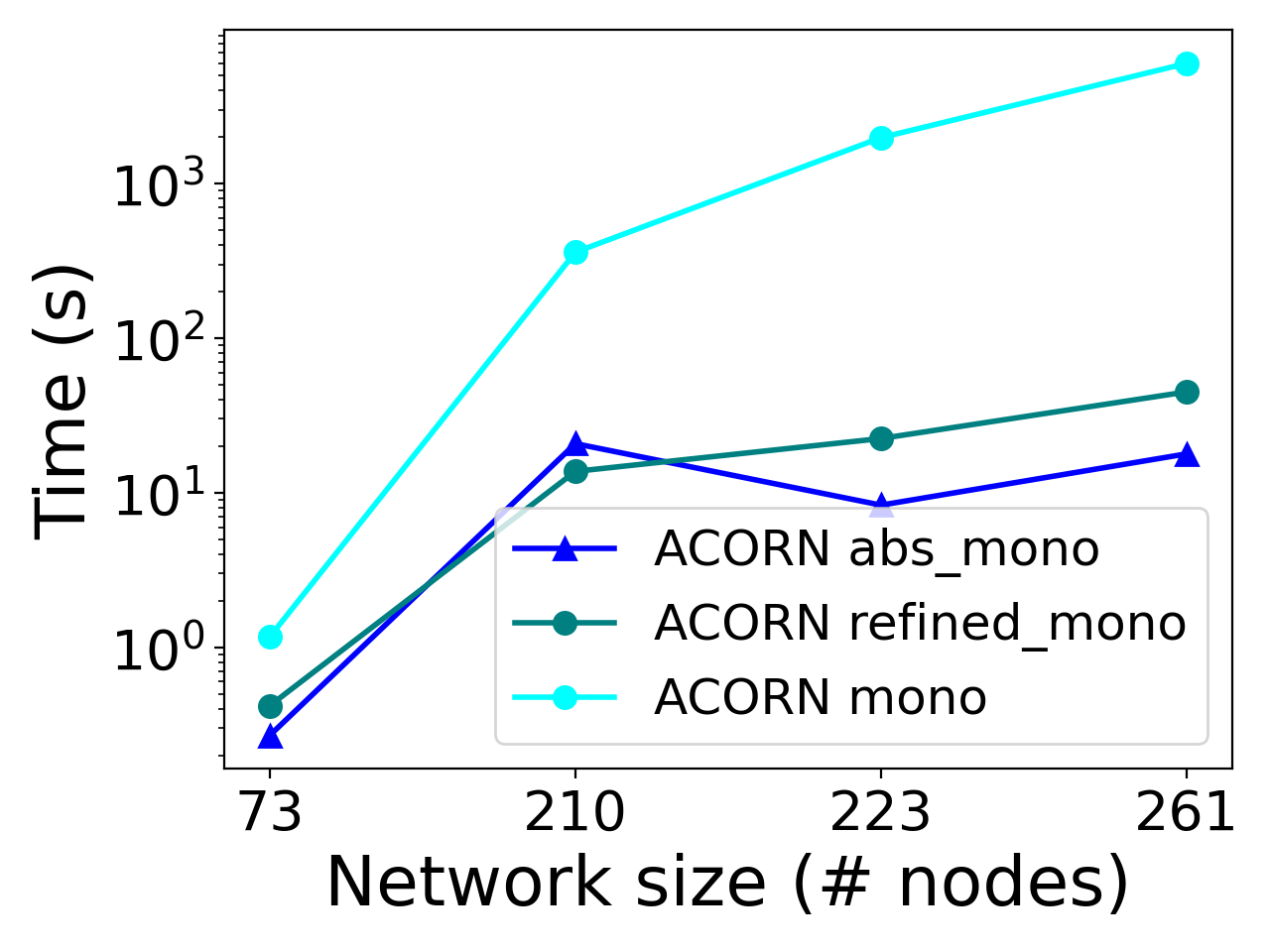}
         \caption{MonoSAT settings}
    \end{subfigure}
    \begin{subfigure}{0.34\textwidth}
        \centering
         \includegraphics[width=\linewidth]{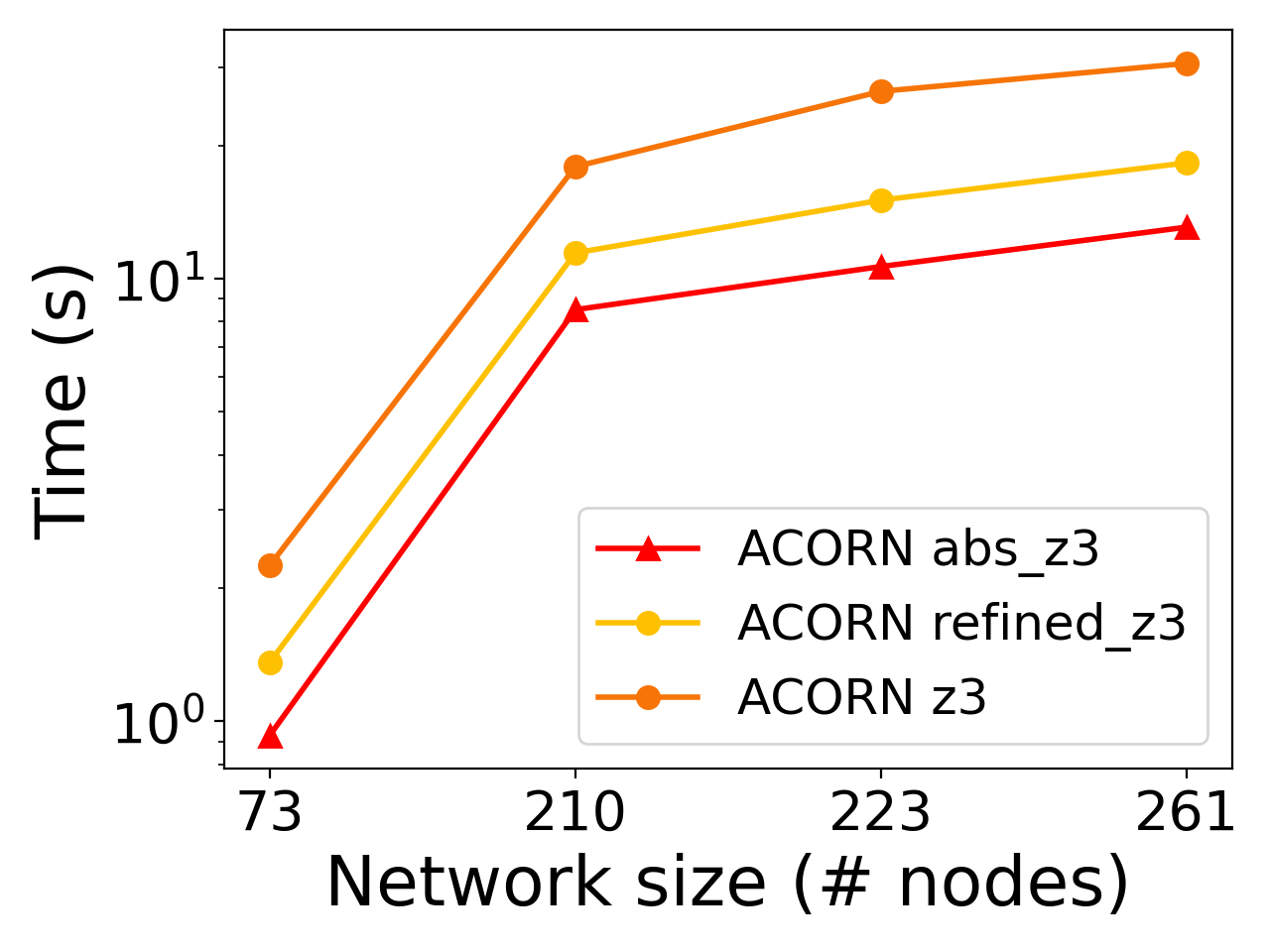}
         \caption{Z3 settings}
    \end{subfigure}
    \caption{Results for refinement of examples with false positives.}
    \label{fig:refinement}
\end{figure}

For the four benchmarks that gave false positives using the  $\prec ^*$ abstraction, we used the $\prec_{(lp)}$ abstraction in which all routers select routes with the highest local preference but without modeling path length and other fields. The results 
for both solvers and all three settings (abstract, refined, and non-abstract) 
are shown in Figure~\ref{fig:refinement}, where the number of nodes is shown on the x-axis, verification time in seconds on the y-axis (log scale) and triangular markers denote false positives. Note that our $\prec_{(lp)}$ abstraction successfully verifies reachability for all four benchmarks, with relative speedups (over no abstraction) of up to 133x for MonoSAT and 1.76x for Z3, and relative slowdowns (over the least precise abstract setting) of up to 2.7x for MonoSAT and 1.5x for Z3. These results demonstrate the precision-cost tradeoff enabled by \sysname. 
\subsection{Comparison with Existing Tools}
\label{s:compare}
To place the performance of \sysname in the context of existing tools, we compared it with two publicly available state-of-the-art control plane verifiers: NV~\cite{nvPldi20} and ShapeShifter~\cite{shapeshifter} (other related tools FastPlane~\cite{fastplane} and Hoyan~\cite{hoyan} are not publicly available). 
ShapeShifter uses abstract interpretation~\cite{ai} to abstract routing messages and implements a fast simulator 
that uses BDDs to represent sets of routing messages. 
NV is a functional programming language for modeling and verifying network control planes, and provides a simulator (based on Multi-Terminal BDDs but without abstraction of routing messages) and 
an SMT-based verifier that uses Z3~\cite{z3}. NV's SMT engine has been shown to perform better than Minesweeper~\cite{nvPldi20}.
NV performs a series of front-end transformations to generate an SMT formula, but its encoding is not based on symbolic graphs. Thus, a comparison of our no-abstraction settings against NV\_SMT indicates the effectiveness of our SMT encodings. For a fair comparison, we only report NV's SMT solving time (and ignore the front-end processing time). 

For the evaluations, we used much larger FatTree topologies ($\approx$ 37,000 nodes) than prior work, with the same common policies -- shortest-path and valley-free routing.
For the data center benchmarks (described earlier, \S\ref{s:dc}), we generated corresponding input formats for ShapeShifter and NV (also publicly available~\cite{benchmark-repo}), such that the routing fields in corresponding inputs are the same for all three tools. (We did not use benchmarks from NV’s repository because its input format is different from our tool. Also, the largest benchmark in the NV repository has 2000 nodes, while we wanted to experiment with larger sizes.)

The results for the shortest-path routing and valley-free policies are shown in 
Figure~\ref{fig:comp} where 
the number of nodes is shown on the x-axis, verification time in seconds on the y-axis (log scale),
timeouts indicated by `x', and out-of-memory indicated by `OOM'. (ShapeShifter and NV could not be applied on the isolation benchmarks as they do not support regular expressions over AS paths.)
\begin{figure*}
    \centering
    \begin{subfigure}{0.35\textwidth}
        \centering
        \includegraphics[width=\linewidth]{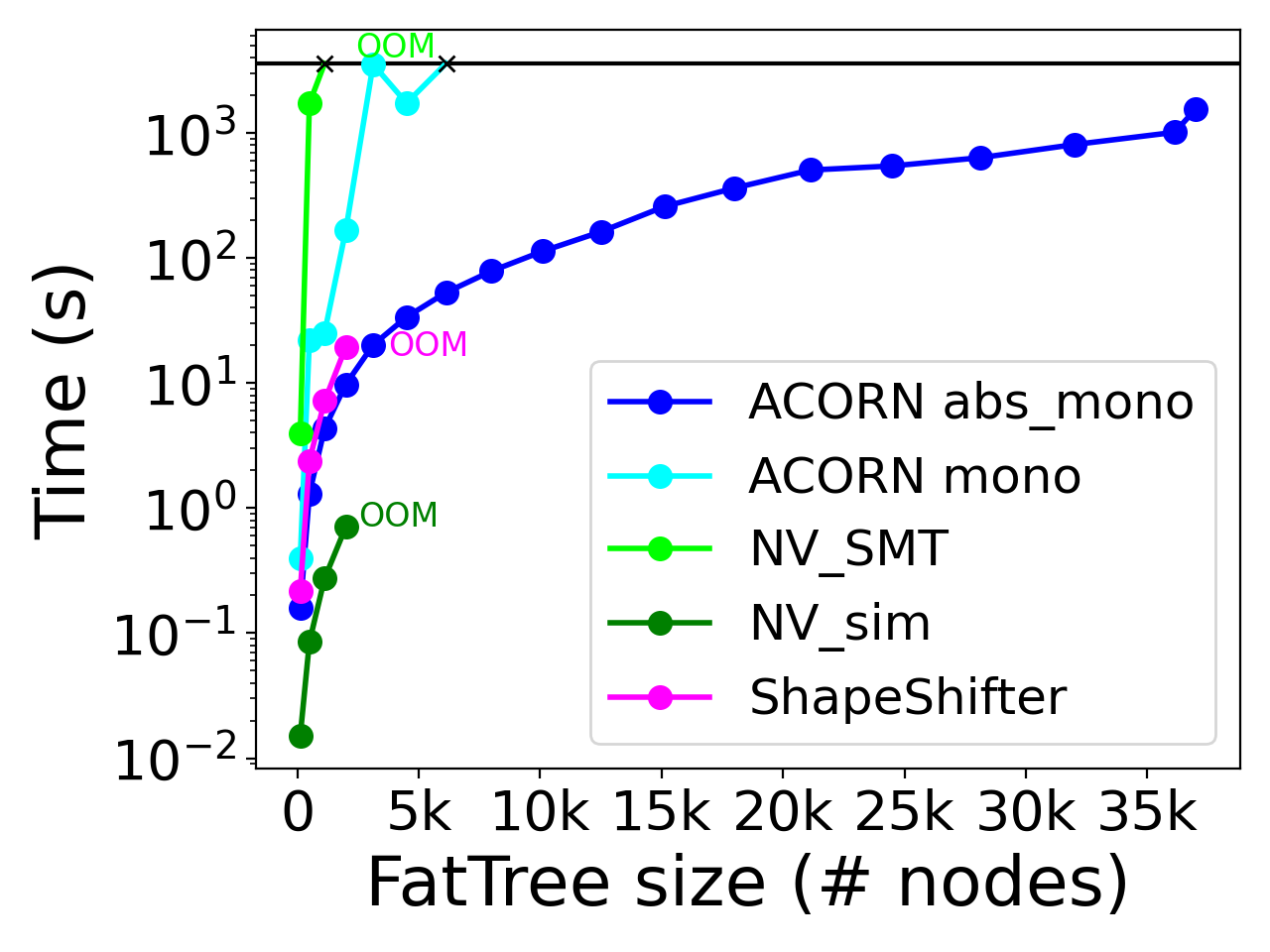}
        \caption{Shortest-path policy}
        \label{fig:comp_sp}
    \end{subfigure}\hspace{0.1\textwidth}
    \begin{subfigure}{0.35\textwidth}
        \centering
         \includegraphics[width=\linewidth]{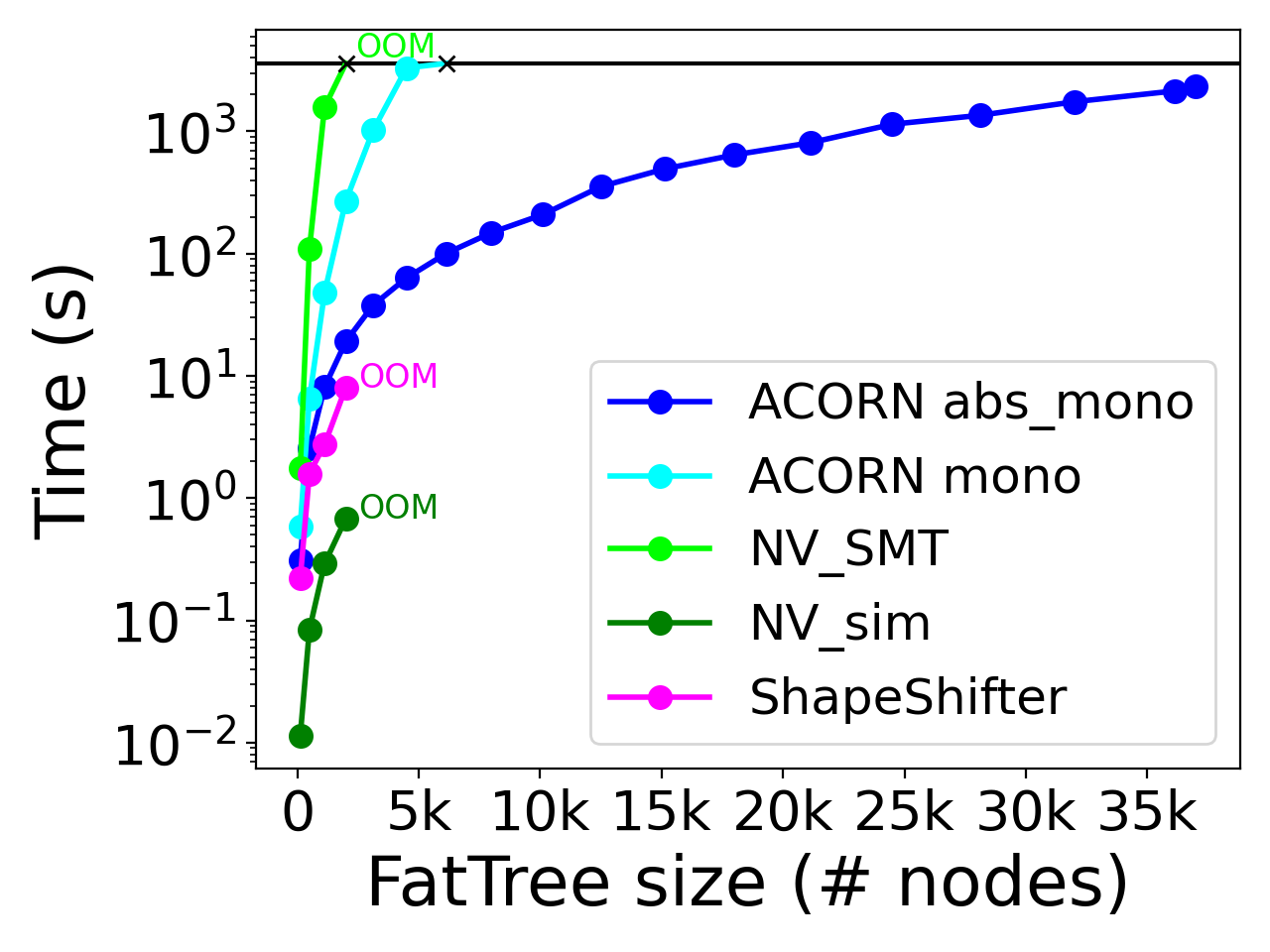}
        \caption{Valley-free policy}
        \label{fig:comp_vf}
    \end{subfigure}
    \caption{Comparison of tools on data center examples
    }
    \label{fig:comp}
\end{figure*}
Note that both NV and ShapeShifter run out of memory for networks with more than 3,000 nodes while \sysname's mono and abs\_mono settings can verify larger networks with 4,500 nodes and 36,980 nodes, respectively. The performance of our abs\_mono setting shows that SMT-based methods for network control plane verification can scale to large networks with tens of thousands of nodes.
        
\subsection{Discussion and Limitations}
\label{s:limitations}
\sysname is sound for properties that hold for all \emph{stable} states of a network, \ie properties of the form $\forall s\ P(s)$ where $s$ is a stable state, such as reachability, policy-based properties, device equivalence, and way-pointing.
Like many SMT-based tools, \sysname cannot verify properties over \emph{transient} states that arise before convergence. For verification to be effective, the selected abstraction should model the fields relevant to the property of interest. For checking reachability, our least precise abstraction works well in practice; to verify a property about the path length between two routers 
a user should 
use an abstraction that models path length (otherwise our verification procedure would give a false positive). We have shown that our abstractions are sound under specified failures; however, our tool does not yet model failures. We plan to extend our SMT encodings to model link/device failures in future work.

\section{Related Work}
\label{s:related}
Our work is related to other efforts in network verification 
and use of nondeterministic abstractions for verification.

\para{Distributed control plane verification} 
These methods~\cite{fsr,bagpipe,era,plankton,minesweeper,nvPldi20} aim to verify all data planes that emerge from the control plane. Simulation-based tools~\cite{batfish,cbgp,fastplane,nvPldi20} work with a concrete environment, i.e., a set of external announcements from neighboring networks. 
Although they are fast and can handle large networks, in practice they can miss errors that are triggered only under certain environments. 

In particular, \textsc{FastPlane}~\cite{fastplane} can scale to large data centers (results shown for $\approx$2000 nodes) by applying a generalized shortest-path algorithm to simulate BGP route selection and propagation. However, it 
requires the network policy to be monotonic, \ie a route announcement's preference decreases on traversing any edge in the network. Our approach does not require the network policy to be monotonic.

Another recent work called \textsc{Hoyan}~\cite{hoyan} uses a hybrid simulation and SMT-based approach. It too keeps track of multiple route announcements received at each router to check reachability under failures, but in the context of the given simulation. Although it does not consider all possible environments, it has been deployed in a real-world WAN (with $O(100)$ nodes) and offers additional  capabilities for finding inconsistencies in network models due to vendor-specific behaviors of devices.

The ShapeShifter~\cite{shapeshifter} work is the closest to ours in terms of route abstractions, but it does not use SMT-based verification and does not scale as well as our tool (\S\ref{s:compare}). 
We allow all routing choices at a node, while ShapeShifter uses a conservative abstraction of the best choice for soundness. 
Allowing all routing choices provides more flexibility in tracking correlated choices across different nodes, much as SMT-based program verification allows path-sensitivity for more precision, in comparison to path-insensitive static analysis. 
As an example, ShapeShifter's ternary abstraction for community tags (which abstracts each community tag bit to $\{0, 1, *\}$) would result in a false positive on Example 2 in Figure~\ref{fig:diamond_example} (\S\ref{s:motiv}), while \sysname verifies it correctly. 

An earlier work, Bagpipe~\cite{bagpipe}, verifies BGP policies using symbolic execution and can verify properties over stable as well as transient states of a network. It uses a simplified BGP route selection procedure that chooses routes with maximum local preference, but does not break ties on other fields such as path length, choosing a route nondeterministically in that case. This is similar to our NRC abstraction using partial order $\prec_{(lp)}$. However, our abstraction hierarchy is more general and can be applied to any routing protocol. 

Other tools have proposed different abstractions and optimizations. 
ARC~\cite{arc} proposed a graph-based abstraction 
and uses graph algorithms to check properties such as reachability and fault tolerance, 
but it does not support protocol features such as local preferences or community tags. Its abstraction has been extended to support quantitative properties in QARC~\cite{qarc}, but with similarly restricted support for protocol features. Tiramisu~\cite{tiramisu} uses a similar graph-based representation as ARC, but with multiple layers to capture inter-protocol dependencies. It was shown to scale better than many state-of-the-art verifiers; however, results were shown only on relatively small networks with up to a few hundred devices. 
Plankton~\cite{plankton} uses explicit-state model checking to exhaustively explore all possible converged states of the control plane and implements several optimizations to reduce the size of the search space. Bonsai~\cite{bonsai} proposed a symmetry-based abstraction to compress the network control plane. However, even when a network topology is symmetric, the network policy could break symmetry, making Bonsai's compression technique less effective. In comparison, our \nrcs do not rely on symmetry. 

There has been some recent work~\cite{kirigami,timepiece,lightyear} on using modular verification techniques to improve the scalability of verification. The core ideas in modular verification are orthogonal to our work, and could potentially be combined with abstractions. Among these efforts, \textsc{Lightyear}\cite{lightyear} also verifies BGP policies using an over-approximation that allows routers to choose any received route -- this corresponds to our NRC abstraction with partial order $\prec^*$. However, unlike our approach, it requires a user to provide suitable invariants.

As far as we know, none of these prior efforts have been shown to verify networks with more than 4,500 devices.

\para{Data plane verification} These efforts~\cite{hsa,anteater,flowchecker,veriflow,netplumber,netkat,zhangAtva13,nod} model the data forwarding rules and check properties such as reachability, absence of routing loops, absence of black holes, etc. Although there are differences in coverage of various network design features, properties, and techniques (symbolic simulation, model checking, SAT/SMT-based queries), many such methods have been shown to successfully handle the scale and complexity of real-world networks. 
Similar to these methods, our least precise abstraction does not model the route selection procedure, 
but we verify all data planes that emerge from the control plane, not just one snapshot.

\para{Nondeterminism and abstractions}
Nondeterminism has been used to abstract behavior in many different settings in software and hardware verification. Examples include control flow nondeterminism 
in Boolean program abstractions in SLAM~\cite{slam}, a sequentialization technique~\cite{LalRepsCav08} that converts control nondeterminism (\ie interleavings in a concurrent program) to data nondeterminism, and 
a localization abstraction~\cite{KurshanLocal} in hardware designs. Our \nrcs
use route nondeterminism to soundly abstract network control plane behavior.

\section{Conclusions and Future Directions}
\label{s:concl}
The main motivation for our work is to provide full symbolic verification of network control planes that can scale to large networks. Our approach is centered around two core contributions:
a hierarchy of nondeterministic abstractions, and a new SMT encoding that can leverage specialized SMT solvers with graph theory support. Our tool, \sysname, 
has verified reachability (an important property for network operators) on data center benchmarks (with FatTree topologies and commonly used policies) with $\approx$37,000 routers, which far exceeds what has been shown by existing related tools. 
Our 
evaluation shows that our abstraction performs \emph{uniformly better} than no abstraction for verifying reachability on different network topologies and policies, and with two different SMT solvers.
In future work, we plan to consider verification under failures,
and combine our abstractions with techniques based on modular verification of network control planes.

\bibliographystyle{splncs04}
\bibliography{references}

\appendix
\section{BGP Overview}
\label{app:bgp}
BGP is the protocol used for routing between \emph{autonomous systems} (ASes) in the Internet.
An autonomous system (AS) is a network controlled by a single administrative entity \eg the network of an Internet Service Provider (ISP) in a particular country, or a college campus network. A simplified version of the decision process used to select best routes in BGP is shown in Table~\ref{tab:bgp_decision_process}~\cite{bgp_policies}. A router compares two route announcements by comparing the attributes in each row of the table, starting from the first row. A route announcement with higher local preference is preferred, regardless of the values of other attributes; if two route announcements have equal local preference, then their path lengths will be compared. BGP allows routes to be associated with additional state via the community attribute, a list of string tags. Decisions can be taken based on the tags present in a route announcement; for example, a route announcement containing a particular tag can be dropped or the route preference can be altered (e.g., by increasing the local preference if a particular tag is present).

\begin{table}[]
    \centering
    \begin{tabular}{|p{0.7cm}|p{2.5cm}|p{4cm}|p{2.5cm}|}
    \hline
        Step & Attribute & Description & Preference (Lower/Higher)  \\
        \hline
        1 & Local preference & An integer set locally and not propagated & Higher \\
        \hline
        2 & AS path length & The number of ASes the route has passed through & Lower \\
        \hline
        3 & Multi-exit Discriminator (MED) & An integer used to influence which link (among many) should be used between two ASes & Lower\\
        \hline
        4 & Router ID & Unique identifier for a router used for tie breaking & Lower\\
        \hline
    \end{tabular}
    \vspace*{0.1in}
    \caption{Simplified BGP decision process to select the best route~\cite{bgp_policies}}
    \label{tab:bgp_decision_process}
\end{table}
\newpage
\section{Proof of soundness of the \nrcs}
\label{app:proofs}

\overapproxlemma*
\begin{proof}
We need to show that for each labeling $\mathcal{L}$, if $\mathcal{L} \in Sol(S)$ then $\mathcal{L} \in Sol(\wh{S}_{\prec '})$. An SRP solution $\mathcal{L}$ is defined by
\begin{align*}
    \mathcal{L}(u) = \begin{cases}
    \mathnormal{a_d} & \text{ if } u = d\\
    \infty & \text{ if } \mathrm{attrs}_{\mathcal{L}}(u) = \emptyset\\
    a \in \mathrm{attrs}_{\mathcal{L}}(u)
    \text{ , minimal by } \prec & \text{ if } \mathrm{attrs}_{\mathcal{L}}(u) \neq \emptyset
    \end{cases}
\end{align*}
where $\mathrm{attrs}_{\mathcal{L}}(u)$ is the set of attributes that $u$ receives from its neighbors. The abstract SRP $\wh{S}_{\prec '}$ differs from the SRP $S$ only in the partial order. Therefore, to show that $\mathcal{L}$ is a solution of $\wh{S}_{\prec '}$, we need to show that if $\mathrm{attrs}_{\mathcal{L}}(u) \neq \emptyset$, then $\mathcal{L}(u)$ is minimal by $\prec '$. By the definition of an abstract SRP, the set of minimal attributes according to $\prec'$ is a superset of the set of minimal attributes according to $\prec$, which means $\mathcal{L}(u)$ is minimal by $\prec '$. Therefore, any SRP solution $\mathcal{L}$ is a solution of the abstract SRP $\wh{S}_{\prec '}$.
\end{proof}

\soundnessthm*
\begin{proof}
If $\wh{N} \land \neg P$ is unsatisfiable, every solution of the abstract SRP satisfies the given property. By Lemma~\ref{lemma:overapprox}, the property also holds for all solutions of the concrete SRP $S$, \ie there is no property violation in the real network.
\end{proof}

\newpage
\section{Transfer Constraints in Abstract SRP}
\label{app:example_transfer_constraints}

\begin{figure}[b]
    \centering
    \includegraphics[width=0.5\linewidth]{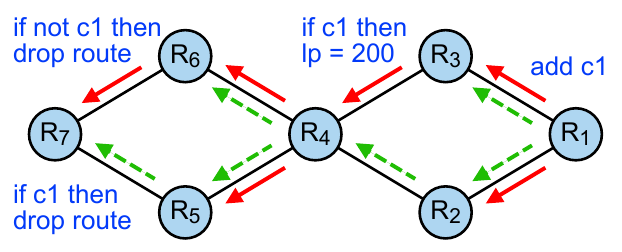}
    \caption{Example 2}
    \label{fig:diamond_example_repeated}
\end{figure}

The attribute transfer and route filtering constraints for the example network shown in Figure~\ref{fig:diamond_example} (reproduced in Figure~\ref{fig:diamond_example_repeated}) are given below. We use a bit vector variable $comm_{u}$ to denote the community attribute at node $R_u$, and $re_{uv}$ to denote the routing edge variable for edge $(R_u, R_v)$. We encode the community tag $c1$ using the value $1$, and use the value $0$ to represent the initial community value at the destination.

\para{Initial route at destination} We set the community to a default value of $0$ at the destination router $R_1$ using the constraint  $comm_{1} = 0$.

\para{Transfer constraints along edge $(R_1, R_3)$} The transfer function along edge $(R_1, R_3)$ updates the community attribute. The route is never dropped along this edge, so the placeholder $routeDropped_{13}$ is False, which makes  equation~\eqref{eq:dropped} in Figure~\ref{fig:encoding} trivially hold. This is encoded using the constraints
\begin{align}
    &re_{13} \rightarrow comm_3 = 1
    \label{eq:trans_13}\\
    &re_{13} \rightarrow \neg routeDropped_{13}\\
    &routeDropped_{13} \leftrightarrow False
\end{align}
Our implementation simplifies formulas wherever possible when $routeDropped$ is a constant, and only asserts equation~\eqref{eq:trans_13} above.

\para{Transfer constraints along edges $(R_5, R_7)$ and $(R_6, R_7)$} The transfer functions along these edges propagate the community attribute and filter routes based on the community. This is encoded using the constraints
\begin{align}
    &re_{57} \rightarrow comm_7 = comm_5\\
    &re_{57} \rightarrow \neg routeDropped_{57}\\
    &routeDropped_{57} \leftrightarrow (comm_5 = 1)
\end{align}
\begin{align}
    &re_{67} \rightarrow comm_7 = comm_6\\
    &re_{67} \rightarrow \neg routeDropped_{67}\\
    &routeDropped_{67} \leftrightarrow (comm_6 \neq 1)
\end{align}

\para{Transfer constraints along other edges} The transfer functions for all other edges propagate the community attribute and do not filter out routes. They are encoded using the constraints
\begin{align}
    &re_{vu} \rightarrow comm_u = comm_v \label{eq:trans_vu}\\
    &re_{vu} \rightarrow \neg routeDropped_{vu} \label{eq:drop_vu}\\
    &routeDropped_{vu} \leftrightarrow False
\end{align}
Since $routeDropped$ is constant for these edges, our implementation would simplify the formulas by substituting the value of $routeDropped$, and would only assert equation~\eqref{eq:trans_vu} above, as equation~\eqref{eq:drop_vu} is trivially true.
\section{\sysname Intermediate Representation (IR) and Benchmark Examples}
\label{app:grammar}
\para{Intermediate Representation (IR)} Our tool implementation takes as input policies written in a simple IR language for the SRP model. The transfer function is represented as a list of match-action rules, similar to route-maps in Cisco's configuration language. We support matching on the community attribute and some types of regular expressions over the AS path. Our implementation currently supports regular expressions that check whether the path contains certain ASes or a particular sequence of ASes, and could be extended to support general regular expressions in the future. A match can be associated with multiple actions, which can update route announcement fields such as the community attribute, local preference, and AS path length.

\para{Benchmark examples} The details of the wide area network examples we used (\S\ref{s:eval_wan}) are described below.\\
\noindent
\emph{Topology Zoo benchmarks.} We used 10 topologies from the Topology Zoo~\cite{topozoo}, which we pre-processed \eg by removing duplicate nodes and nodes with id ``None". The names and sizes of the resulting topologies are shown in Table~\ref{tab:tz_examples}. We annotated the topologies with business relationships between ASes, considering each node as an AS, and used a BGP policy that implements the Gao-Rexford conditions~\cite{gao-rexford}. The annotated benchmark files (in GML format) are included in our benchmark repository~\cite{benchmark-repo}, along with the examples in our IR format.\\
\noindent
\begin{table}[]
    \centering
    \begin{tabular}{|c|c|c|}
    \hline
        Benchmark & Topology name & Size \\
    \hline
        TZ1 & VinaREN & 22 nodes, 24 edges\\
        TZ2 & FCCN & 23 nodes, 25 edges\\
        TZ3 & GTS Hungary & 27 nodes, 28 edges\\
        TZ4 & GTS Slovakia & 32 nodes, 34 edges\\
        TZ5 & GRnet & 36 nodes, 41 edges\\
        TZ6 & RoEduNet & 41 nodes, 45 edges\\
        TZ7 & LITNET & 42 nodes, 42 edges\\
        TZ8 & Bell South & 47 nodes, 62 edges\\
        TZ9 & Tecove & 70 nodes, 70 edges\\
        TZ10 & ULAKNET & 79 nodes, 79 edges\\
    \hline
    \end{tabular}
    \caption{Topology Zoo examples}
    \label{tab:tz_examples}
\end{table}
\noindent
\emph{BGPStream benchmarks.} We created a set of 10 examples based on parts of the Internet involved in BGP hijacking incidents, as reported on BGPStream~\cite{bgpstream}. For a given BGP hijacking incident, we created a network with the ASes involved and used the CAIDA AS Relationships dataset~\cite{caida} to add edges between ASes with the given business relationships (customer-provider or peer-peer). We then removed some additional ASes (if required) so that our no-abstraction setting could verify that all ASes in the resulting network can reach the destination (taken to be the possibly hijacked AS). We used a BGP policy that implements the Gao-Rexford conditions~\cite{gao-rexford} based on the given business relationships between ASes. The details of these examples are shown in Table~\ref{tab:bgpstream_examples}.

\begin{table*}[]
    \centering
    \begin{tabular}{|c|c|c|}
    \hline
    Benchmark & Incident date & Size \\
    \hline
    B1 & 2021-06-14 & 261 nodes, 3325 edges\\
    B2 & 2021-06-17 & 223 nodes, 2722 edges\\
    B3 & 2021-06-18 & 133 nodes, 1205 edges\\
    B4 & 2021-06-19 & 210 nodes, 2100 edges\\
    B5 & 2021-06-21 & 269 nodes, 3351 edges\\
    B6 & 2021-06-22 & 212 nodes, 2233 edges\\
    B7 & 2021-06-22 & 294 nodes, 4108 edges\\
    B8 & 2021-06-22 & 124 nodes, 860 edges\\
    B9 & 2021-06-22 & 73 nodes, 270 edges\\
    B10 & 2021-06-25 & 154 nodes, 1176 edges\\
    \hline
    \end{tabular}
    \caption{BGPStream examples}
    \label{tab:bgpstream_examples}
\end{table*}

\para{Results for Topology Zoo examples} Detailed results for the Topology Zoo benchmark examples are shown in Figure~\ref{fig:topo_zoo_results}. All settings take less than 0.5s to verify each property.

\begin{figure*}
   \centering
\begin{subfigure}{0.45\textwidth}
        \centering
        \includegraphics[width=\linewidth]{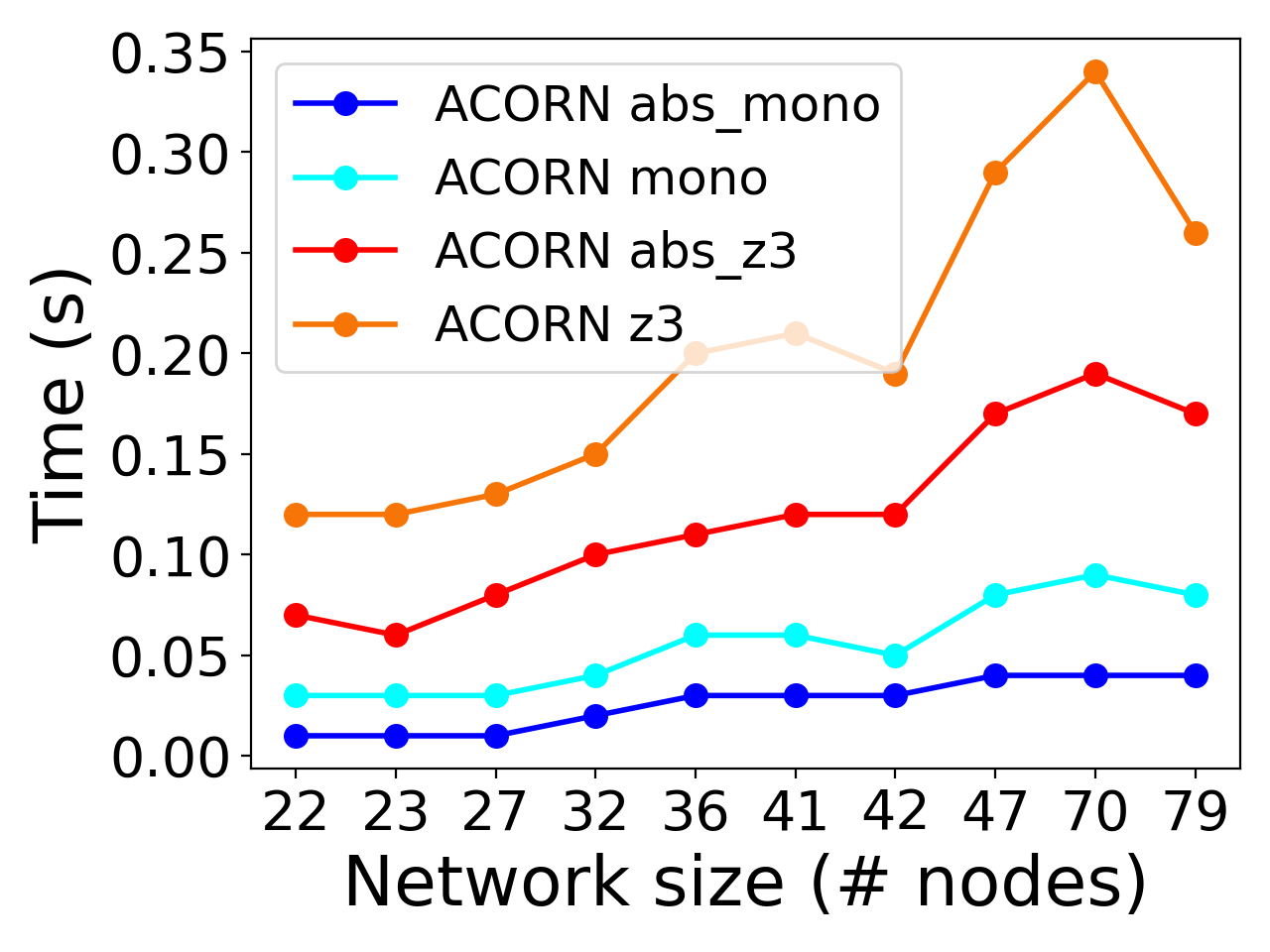}
        \caption{Reachability}
        \label{fig:topo_zoo_reach}
\end{subfigure}
\begin{subfigure}{0.45\textwidth}
    \centering
    \includegraphics[width=\linewidth]{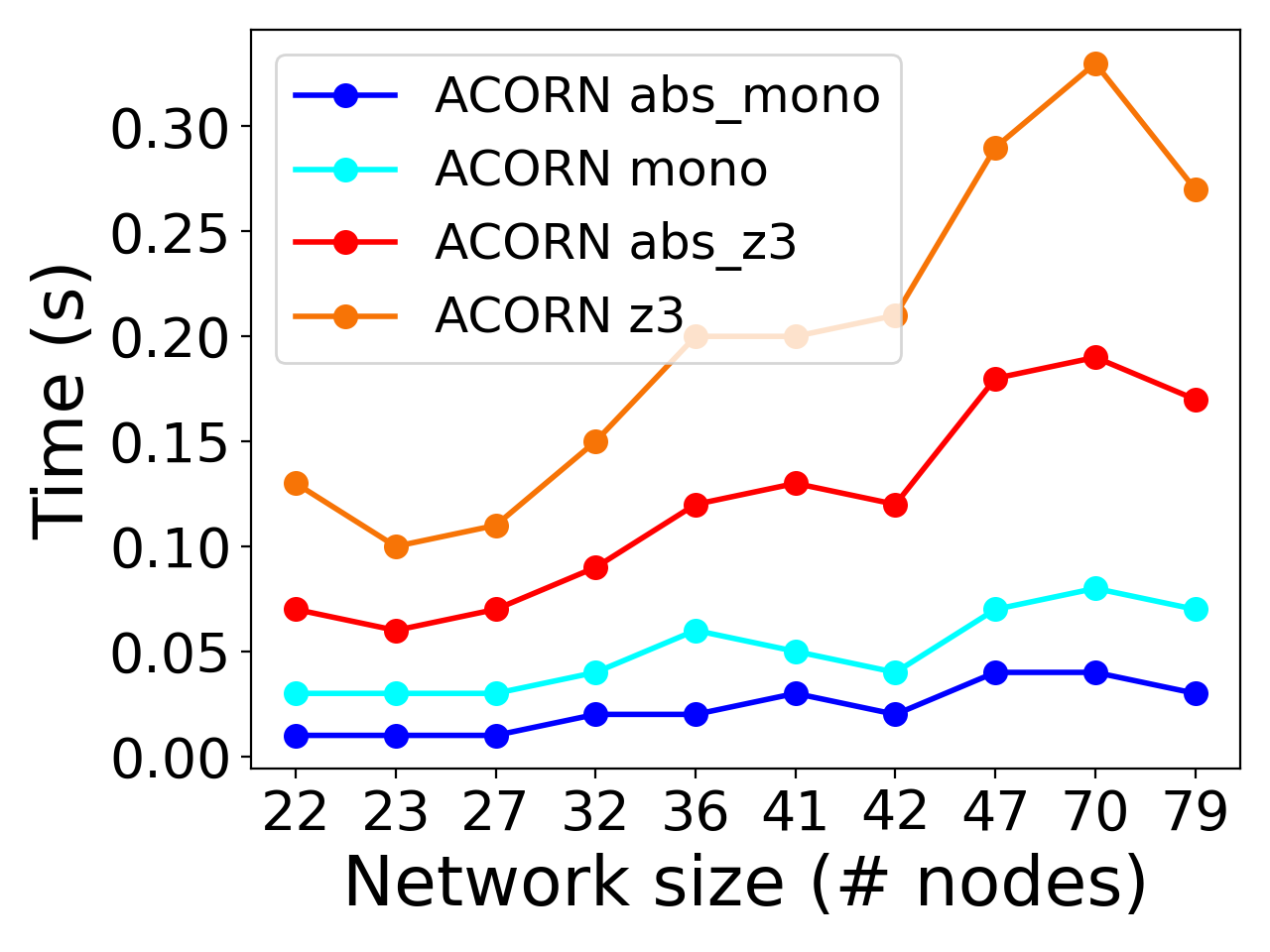}
    \caption{No-transit property}
    \label{fig:topo_zoo_no_transit}
\end{subfigure}
\caption{Results for Topology Zoo examples}
\label{fig:topo_zoo_results}
\end{figure*}

\section{SMT Constraints for Concrete SRP}
\label{app:concreteSmt}
To encode a concrete SRP we add constraints to ensure that each node picks the best route, in addition to the constraints used to encode an abstract SRP (Figure~\ref{fig:encoding}). We ensure that for every edge $(v, u) \in E$, if $u$ selects the route from $v$ (if $nChoice_u = nID(u, v)$) then $v$'s route must be the best route that $u$ receives from its neighbors. This will require us to keep track of the attribute fields used in the route selection procedure (such as path length) and possibly additional temporary variables to track the minimum or maximum value of an attribute. We illustrate how our abstract formulation can be extended to a concrete formulation by showing the constraints required to model the first two steps in BGP's route selection procedure.

\begin{example}[Encoding route selection in BGP]
To encode the first two steps in BGP's route selection procedure we need to keep track of two additional attributes: local preference (denoted $lp$) and the AS path length (denoted $path$), and encode transfer constraints for these attributes (such as constraints which increment the path length along an edge). We then add constraints that model the route selection procedure.
\end{example}
We introduce one variable per edge $(v, u)$, $trans\_lp_{vu}$, which denotes the local preference of the route announcement sent from $v$ to $u$ after applying the transfer function along the edge. We also introduce two variables for each node: $maxLp$, which tracks the maximum $lp$ of routes received from its neighbors; and $minPath$, which tracks the minimum path length among routes with the maximum local preference.
The variable $maxLp_u$ is defined by the constraints
\begin{align*}
    &\bigwedge_{(v, u) \in E} nValid_{vu} \rightarrow maxLp_u \geq trans\_lp_{vu}\\
    &nChoice_u \neq None_u \rightarrow
    \bigvee_{(v, u) \in E} nValid_{vu} \land maxLp_u = trans\_lp_{vu}\\
\end{align*}
where $ nValid_{vu} \leftrightarrow (hasRoute_v \land \neg routeDropped_{vu})$ indicates whether $v$ sends a route to $u$. To set the minimum path length, we use similar constraints:
\begin{align*}
    &\bigwedge_{(v, u) \in E} (nValid_{vu} \land trans\_lp_{vu} = maxLp_u) \rightarrow  minPath_u \leq path_v\\
    &nChoice_u \neq None_u \rightarrow\\
    &\bigvee_{(v, u) \in E} nValid_{vu} \land trans\_lp_{vu} = maxLp_u \land minPath_u = path_v\\
\end{align*}
We now add constraints on the $nChoice$ variables at each node to ensure that if $u$ chooses a route from any neighbor $v$, $v$'s route must be the best.
\begin{align*}
    nChoice_u = nID(u, v) \rightarrow trans\_lp_{vu} = maxLp_u \land path_v = minPath_u
\end{align*}

\end{document}